\documentclass[smallextended]{svjour3}       

\usepackage{newtxtext,newtxmath}
\usepackage{hyperref}
\usepackage{amsfonts, amssymb}
\usepackage{physics}
\usepackage{enumitem}
\usepackage{natbib}
\newcommand{\kket}[1]{\|{#1}\rangle\!\rangle}
\newcommand{\kkket}[1]{|\hspace{-0.1em}\|{#1}\rangle\!\rangle\!\rangle} 
\usepackage{graphicx}
\usepackage[justification=justified,singlelinecheck=false, width=.5\textwidth, format=plain]{caption}
\usepackage{framed}
\usepackage{mathrsfs}

\DeclareMathOperator{\supp}{supp}
\spnewtheorem{A}{A}{\bf}{\rm}
 \journalname{Manuscript submitted to Foundations of Physics}
\begin{document}

\title{Quantum information vs.\ epistemic logic: An analysis of the Frauchiger-Renner theorem\thanks{Some amount of work for this paper was done during my employment with the research unit The Epistemology of the Large Hadron Collider, funded by the German Research Foundation (DFG; grant FOR 2063). The work was initiated with my visiting scholarship at the Stellenbosch Institute for Advanced Study, kindly hosted by the latter.}
}

\titlerunning{Quantum information vs.\ epistemic logic}        

\author{Florian J. Boge
}


\institute{Florian J. Boge \\[.5em]  Institute for Theoretical Particle Physics and Cosmology\\ 
RWTH Aachen University\\ Sommerfeldstr. 16, 52074 Aachen, Germany \\[.5em] 
Stellenbosch Institute for Advanced Study (STIAS)\\
Wallenberg Research Centre at Stellenbosch University\\
Stellenbosch 7600\\
South Africa\\[.5em] 
fboge@physik.rwth-aachen.de 
}

\date{Received: date / Accepted: date}

\maketitle

\begin{abstract}
A recent no-go theorem \citep[][]{frauchiger} establishes a contradiction from a specific application of quantum theory to a multi-agent setting. The proof of this theorem relies heavily on notions such as `knows' or `is certain that'. This has stimulated an analysis of the theorem by Nurgalieva and del Rio \cite{delrio}, in which they claim that it shows the ``[i]nadequacy of modal logic in quantum settings'' (ibid.). In this paper, we will offer a significantly extended and refined reconstruction of the theorem in multi-agent modal logic. We will then show that a thorough reconstruction of the proof as given by Frauchiger and Renner requires the reflexivity of access relations (system \textsf{\textbf{T}}). However, a stronger theorem is possible that already follows in serial frames, and hence also holds in systems of \emph{doxastic} logic (such as \textsf{\textbf{KD45}}). After proving this, we will discuss the general implications for different interpretations of quantum probabilities as well as several options for dealing with the result.\end{abstract}

\section{Introduction}\label{sec:Intro}
Quantum theory (QT) is the most successful theory of modern science \citep[e.g.][]{kleppner}, and yet notoriously defies even a more-or-less agreed-upon interpretation. In the mathematical formalism, to recall, a system's state is modelled by a \emph{density operator} $\rho$, a positive trace class operator with trace 1 on some separable Hilbert space $\mathcal{H}$, i.e., $\ev{\rho}{v}=\braket*{v}{\rho v}\geq 0$ for all $\ket{v}\in\mathcal{H}$, where $\braket{\cdot}$ defines a Hermitian form on $\mathcal{H}$, and $\tr(\rho)=_{df}\sum_{i=1}^{d}\ev{\rho}{v_{i}}=1$ ($\qty{\ket{v_{i}}}_{1\leq i\leq d}$ an orthonormal basis of $\mathcal{H}$, and $\dim \mathcal{H}=d\in\mathbb{R}\cup\qty{\infty}$). 

In case $\rho$ corresponds to a \emph{projector} $\pi_{v}=\dyad*{v}$, it is called \emph{pure} and in turn corresponds to an equivalence class of vectors $\ket*{v}$ from $\mathcal{H}$ with $\ket*{v}\sim\ket*{v}'$ iff (if, and only if) there is a $\phi\in\mathbb{R}$ s.t.\ $\ket*{v'}=e^{\imath \phi}\ket*{v}$ ($\imath=\sqrt{-1}$). $\pi_{v}$ is positive and satisfies $\pi_{v}=\pi_{v}^{\dagger}$ and $\pi_{v}^{2}=\pi_{v}$, where the \emph{adjoint} $X^{\dagger}$ is defined by $\braket*{v}{X^{\dagger}v'}=\braket*{Xv}{v'}$ ($\ket*{Yv}$ the vector resulting from applying operator $Y$ to $\ket*{v}$ and $\bra*{v}$ an element of the dual space $\mathcal{H}^{\ast}$ of $\mathcal{H}$). Non-pure states, also called \emph{mixed}, are convex sums $\rho =\sum_{k}\lambda_k \pi_k$, where the $\lambda_k\geq 0$ satisfy $\sum_k\lambda_k =1$ and can be interpreted as probabilities. 

According to the \emph{Sch\"odinger picture}, the state of the system evolves as $U\rho U^{\dagger}$, were $U$ is a bijective, linear, norm-preserving (\emph{unitary}) map on $\mathcal{H}$ that depends on a time-parameter $t$, and $UU^{\dagger}=U^{\dagger}U$ gives the identity map $\mathbb{I}$ on $\mathcal{H}$. (For $N$ spaces $\mathcal{H}_j$, we let $\mathbb{I}_N$ the identity on $\bigotimes_{j=1}^{N} \mathcal{H}_j$.) The simplest kind of measurement on a system $\mathfrak{s}$ with state $\rho_\mathfrak{s}$ is given by a family of one dimensional orthogonal projectors $\qty{\pi_{j}}_{1\leq j \leq d}$, each corresponding to a one dimensional subspace of $\mathcal{H}$, where orthogonality means $\pi_i\pi_j=\delta_{ij}\mathbb{I}$ ($\delta_{ij}$ the Kronecker symbol). Equivalently, the \emph{Heisenberg picture} has the $\pi_{j}$ evolve over time as $\pi_{j}^{t}=U^{\dagger}_{t}\pi_{j}U_{t}$.

According to \emph{Born's rule}, the probability for obtaining outcome $o$ in a corresponding measurement on $\mathfrak{s}$ is given by $\tr(\rho\pi_{o})$, which simplifies to $\ev{\pi_{o}}{v}=\braket*{v}{o}\!\braket*{o}{v}$ for a pure state ($\rho=\dyad*{v}$). One may define an \emph{observable} $O$ whose values correspond to these outcomes either directly by the set of projectors or by $O=\sum_{o}o\pi_{o}$. The set $\sigma(O)$ of values $o$ of $O$ is also called $O$'s eigenvalue-spectrum, and Born's rule defines a probability measure $\Pr(O=o)=\tr(\rho \pi_o)$ over $\sigma(O)$.

Measurements of this kind are called \emph{projective}, and if one particular outcome results, they are additionally called \emph{selective}. The state update-prescription for this kind of measurement, called the \emph{L\"udes rule}, has $\rho$ evolve as
\begin{equation}\label{eq:lud}
\rho\mapsto\pi_{o}\rho\pi_{o}/\tr(\rho\pi_{o}),
\end{equation}
where $\pi_{o}$ corresponds to the value $o$ of $O$ measured. 

On account of the evolution described in \eqref{eq:lud}, the state of the system measured will become $\pi_{o}$, which corresponds to the \emph{eigenstate} $\ket*{o}$ of $O$, meaning that it satisfies $O\ket*{o}=o\ket*{o}$. However, even assuming $\rho$ to be pure ($\rho=\dyad*{v}$), it could have a (non-trivial) expansion $\ket{v}=\sum_{o}\alpha_{o}\ket{o}$ in terms of $O$'s eigenstates, where $\alpha_{o}\in\mathbb{C}$ for each $o$. 

The L\"uders step $\pi_v\mapsto\pi_o$ is thus generally incompatible with the unitary evolution: Assuming $\ket{\psi}_\mathfrak{s}=\ket{o}_\mathfrak{s}$ for some system system $\mathfrak{s}$ being measured, the joint evolution of $\mathfrak{s}$ and measuring device $\mathfrak{m}$ should, in the simplest case called an \emph{ideal von Neumann measurement}, proceed as $U\ket*{o}_ \mathfrak{s}\ket*{0}_\mathfrak{m}=\ket*{o}_ \mathfrak{s}\ket*{o}_\mathfrak{m}$, where $\ket*{0}_\mathfrak{m}$ is some state in which $\mathfrak{m}$ does not indicate anything particular about $\mathfrak{s}$'s state, and $\ket*{o}_\mathfrak{m}$ one where it indicates $\mathfrak{s}$'s state exactly. However, if the state of $\mathfrak{s}$ is $\ket{v}_\mathfrak{s}=\sum_o\alpha_o\ket*{o}_\mathfrak{s}$, the same unitary will effect the change $U\sum_i\alpha_i\ket*{o}_\mathfrak{s}\ket*{0}_\mathfrak{m}=\sum_o\alpha_o\ket*{o}_\mathfrak{s}\ket*{o}_\mathfrak{m}$. The incompatibility between these two kinds of processes is known as the \emph{quantum measurement problem} \citep[e.g.][]{maudlin1995}.

The implications of the measurement problem are far reaching and well-known, in part on account of thought experiments such as that of \emph{Wigner's friend} \citep[][]{wigner}: Assume that a friend $\mathfrak{f}$ of physicist Eugene Wigner is measuring an observable of some quantum system $\mathfrak{s}$ inside the confines of some laboratory. $\mathfrak{f}$ then finds outcome $o$ and assigns state $\rho=\pi_{o}$ accordingly. Now assume also that Wigner is outside the laboratory and wants to find out the joint state of $\mathfrak{f}$ and $\mathfrak{s}$ by making an appropriate measurement. In a sense, the friend is a physical system: organisms consist of molecules, which consist of atoms, which consist of elementary particles. Hence, the joint system $\mathfrak{sf}$ could be associated with a however complex density operator $\rho_{\mathfrak{sf}}$ that evolves unitarily in time.

From Wigner's point of view, neither should $\mathfrak{s}$ therefore be in a state $\pi_{o}$ at the time the friend chose to measure it, nor should $\mathfrak{f}$ himself be in an eigenstate of most of the observables Wigner could in principle measure on him. To avoid the conclusion that $\mathfrak{f}$ would hence end up ``in a state of suspended animation'', Wigner attributed to ``consciousness[...] a different role in quantum mechanics than [to] the inanimate measuring device[...].'' \citep[293]{wigner} 

Recently, Frauchiger and Renner \cite{frauchiger} have extended the scenario of Wigner's friend, the result being a theorem that purports to demonstrate the inconsistency of single outcome interpretations of QT in which it can be universally applied. The scenario concerns ``multiple agents'' that each ``have access to different pieces of information, and draw conclusions by reasoning about the information held by others'' \citep[6]{frauchiger}, so it falls in the domain of \emph{quantum information theory}, ``the study of the information processing tasks [...] accomplished using quantum mechanical systems.'' \citep[1]{nielsen2010} Reasoning about information in multi-agent scenarios is also the proper subject of \emph{epistemic logic} \citep[e.g.][1]{ditmarsch}. Hence, it should be possible to reconstruct the theorem within epistemic logic---as has been done by Nurgalieva and del Rio \cite{delrio}. 

We here significantly extend and refine Nurgalieva and del Rio's approach to better understand the consequences of Frauchiger and Renner's theorem. We will also show that the theorem follows in epistemic as well as doxastic systems, and discuss implications and options for dealing with them.

\section{Intuitive approach}\label{sec:FRThm}
Frauchiger and Renner \cite[2]{frauchiger} consider the following experimental protocol, involving four agents -- call them Amanda ($\mathfrak{a}$), Chris ($\mathfrak{c}$), David ($\mathfrak{d}$), and Gustavo ($\mathfrak{g}$) -- and two \emph{qbits} (two-state quantum systems) $\mathfrak{r}, \mathfrak{l}$:
\begin{framed}
\textsf{FR-Protocol:}
\begin{itemize}[wide=0pt, leftmargin=15pt, labelwidth=15pt, align=left]
\item[$t=0$:] The initial state of $\mathfrak{r}$ is $\ket*{\text{init}}_\mathfrak{r}=\sqrt{1/3}\ket*{0}_{\mathfrak{r}}+\sqrt{2/3}\ket*{1}_{\mathfrak{r}}$, and that of $\mathfrak{l}$ is $\ket*{0}_{\mathfrak{l}}$. The memories of $\mathfrak{a},\mathfrak{c},\mathfrak{d},\mathfrak{g}$ are each associated with quantum states $\ket*{0}_{j} (j\in\qty{\mathfrak{a},\mathfrak{c},\mathfrak{d},\mathfrak{g}})$.
\item[$t=1$:] $\mathfrak{a}$ performs a measurement $\qty{\pi_{0},\pi_{1}}$ on $\mathfrak{r}$ and her memory is updated to $\ket*{0}_\mathfrak{a}$ or $\ket*{1}_\mathfrak{a}$ respectively, depending on the result. If she measures $0$, she leaves $\mathfrak{l}$ in its initial state; otherwise she prepares it in state $\ket*{+}_\mathfrak{l}=\sqrt{1/2}(\ket*{0}_{\mathfrak{l}}  + \ket*{1}_{\mathfrak{l}})$ by an appropriate procedure. Afterwards, she sends it to $\mathfrak{g}$ in each case.
\item[$t=2$:] $\mathfrak{g}$ performs a measurement $\qty{\pi_{0},\pi_{1}}$ on $\mathfrak{l}$ and his memory is updated accordingly.
\item[$t=3$:] $\mathfrak{c}$ performs a joint measurement $\qty{\pi_{\text{ok}}, \pi_{\text{fail}}}$ on $\mathfrak{r}$ and $\mathfrak{a}$'s memory, where
\begin{align*}
\ket*{\text{ok}}_{\mathfrak{ra}} = \sqrt{1/2}(\ket*{0}_\mathfrak{r}\ket*{0}_\mathfrak{a} - \ket*{1}_\mathfrak{r}\ket*{1}_\mathfrak{a}),\\
\ket*{\text{fail}}_{\mathfrak{ra}} = \sqrt{1/2}(\ket*{0}_\mathfrak{r}\ket*{0}_\mathfrak{a} + \ket*{1}_\mathfrak{r}\ket*{1}_\mathfrak{a}),
\end{align*}
and his memory is updated accordingly.
\item[$t=4$:] $\mathfrak{d}$ performs a joint measurement $\qty{\pi_{\text{ok}}, \pi_{\text{fail}}}$ on $\mathfrak{l}$ and $\mathfrak{g}$'s memory, where
\begin{align*}
\ket*{\text{ok}}_{\mathfrak{lg}} = \sqrt{1/2}(\ket*{0}_\mathfrak{l}\ket*{0}_\mathfrak{g} - \ket*{1}_\mathfrak{l}\ket*{1}_\mathfrak{g}),\\
\ket*{\text{fail}}_{\mathfrak{lg}} = \sqrt{1/2}(\ket*{0}_\mathfrak{l}\ket*{0}_\mathfrak{g} + \ket*{1}_\mathfrak{l}\ket*{1}_\mathfrak{g}),
\end{align*}
and his memory is updated accordingly.
\item[$t=5$:] $\mathfrak{c}$ and $\mathfrak{d}$ compare their outcomes and halt the experiment if they both have measured `ok' on their respective system. Otherwise they set $t=0$ and the protocol repeats. 
\end{itemize} 
\end{framed}

The unitaries relevant for the physical processes involved in the \textsf{FR-Protocol} are given in the \hyperref[sec:append]{Appendix}, \hyperref[thm:OkZero]{A 1}. Allowing that agents can treat other agents as physical systems and describe their co-evolution with other systems unitarily, one obtains
\begin{align*}
 & (\sqrt{1/3}\ket*{0}_{\mathfrak{r}}\ket*{0}_\mathfrak{a}\ket*{0}_\mathfrak{l}+\sqrt{2/3}\ket*{1}_{\mathfrak{r}}\ket*{1}_\mathfrak{a}\sqrt{1/2}(\ket*{0}_\mathfrak{l}+\ket*{1}_\mathfrak{l}))\ket*{0}_\mathfrak{g} \\
=  & (\sqrt{1/3}(\ket*{0}_{\mathfrak{r}}\ket*{0}_\mathfrak{a} +  \ket*{1}_{\mathfrak{r}}\ket*{1}_\mathfrak{a})\ket*{0}_\mathfrak{l}+ \sqrt{1/3}\ket*{1}_{\mathfrak{r}}\ket*{1}_\mathfrak{a}\ket*{1}_\mathfrak{l})\ket*{0}_\mathfrak{g} \\
= &  (\sqrt{2/3}\ket*{\text{fail}}_{\mathfrak{ra}}\ket*{0}_\mathfrak{l}+ \sqrt{1/3}\ket*{1}_{\mathfrak{r}}\ket*{1}_\mathfrak{a}\ket*{1}_\mathfrak{l})\ket*{0}_\mathfrak{g}
\end{align*}
at $t'>t_1>1$, where $t_j$ generically refers to the time in interval $(j, j+1]$ at which the given measurement is completed and $t'$ denotes the time when Amanda has finished her subsequent preparation on $\mathfrak{l}$.\footnote{We consider all measurement `pointlike' and time $t'$ of Amanda's preparation to be `just after' her measurement.} Similarly, after $t_2$, the total state will be 
\begin{align}
& \sqrt{2/3}\ket*{\text{fail}}_{\mathfrak{ra}}\ket*{0}_\mathfrak{l}\ket*{0}_\mathfrak{g}+ \sqrt{1/3}\ket*{1}_{\mathfrak{r}}\ket*{1}_\mathfrak{a}\ket*{1}_\mathfrak{l}\ket*{1}_\mathfrak{g}\label{eq:Chris}\\
= & \sqrt{1/3}\ket*{0}_{\mathfrak{r}}\ket*{0}_\mathfrak{a}\ket*{0}_\mathfrak{l}\ket*{0}_\mathfrak{g}+\sqrt{2/3}\ket*{1}_\mathfrak{r}\ket*{1}_\mathfrak{a}\ket*{\text{fail}}_{\mathfrak{lg}}=:\ket*{\Psi}_{\mathfrak{ralg}}.\label{eq:David}
\end{align}

By Born's rule, 
\begin{equation*}
\Pr(M_{\mathfrak{c}}^{t_3}=\text{ok}\wedge M_{\mathfrak{d}}^{t_4}=\text{ok})={}_{\mathfrak{ralg}}\!\bra{\Psi}\pi_{\text{ok}}\otimes\pi_{\text{ok}}\ket*{\Psi}_{\mathfrak{ralg}}=1/12,
\end{equation*}
so QT renders it expected that the experiment may come to a halt at some point. Assuming that this has happened, what could David infer from QT? 

To see how his conclusion leads to a conflict, consider the following reasoning principles invoked by \cite[4--5]{frauchiger}: 
\begin{itemize}
\item[\textsf{Q}] Suppose agent $\mathfrak{x}$ has established: (i) System $\mathfrak{s}$ is in state $\ket*{v}_{\mathfrak{s}}$ at $t_{0}$. Suppose furthermore that $\mathfrak{x}$ knows: (ii) A value for $O$ is obtained by a measurement of the projectors $\qty{\pi_{o}^{t_{0}}}_{o\in\sigma(O)}$ on $\mathfrak{s}$ at $t_{0}$, completed at $t$. If $_\mathfrak{s}\!\ev{\pi_{\tilde{o}}^{t_{0}}}{v}_\mathfrak{s}=1$ for some $\tilde{o}\in\sigma(O)$, then $\mathfrak{x}$ can conclude: (iii) I am certain that $O=\tilde{o}$ at $t$.
\end{itemize}

\begin{itemize}
\item[\textsf{C}] Suppose agent $\mathfrak{x}$ has established: (i) I am certain that agent $\mathfrak{y}$, upon reasoning with the same theory as I, is certain that  $O=\tilde{o}$ at $t$. Then $\mathfrak{x}$ can conclude: (ii) I am certain that $O=\tilde{o}$ at $t$. \end{itemize}

\begin{itemize}
\item[\textsf{S}] Suppose agent $\mathfrak{x}$ has established: (i) I am certain that $O=\tilde{o}$ at $t$. Then $\mathfrak{x}$ must deny: (ii) I am certain that $O\neq\tilde{o}$ at $t$.\end{itemize}

\textsf{S} may be understood as expressing a `single world' view of QT, as in a many worlds-view, $O$ may be regarded as simultaneously taking on multiple, non-identical values, confined (after decoherence) to more or less independent `worlds' or `branches'. \textsf{Q} establishes a semantic link between quantum probabilities and value statements that may be regarded certain. \textsf{C}, finally, means a removal of iterations of `certainty', regardless of whose agent's certainties they are.

Finally, it is an additional assumption \citep[cf.][272]{delrio} that observables defined over joint spaces can be evolved unitarily:\footnote{Frauchiger and Renner appeal to isometries, linear and norm-preserving but not necessarily bijective maps. Like Nurgalieva and del Rio, however, we assume unitary extensions without loss of generality.}
\begin{itemize}
\item[\textsf{U}] Suppose agent $\mathfrak{x}$ has established: (i) System $\mathfrak{s}\neq\mathfrak{x}$ is in state $\ket*{v}_{\mathfrak{s}}$ at time $t_{0}$. Suppose furthermore that $\mathfrak{x}$ knows: (ii) A value for $O$ is obtained by a measurement of the projectors $\qty{\pi_{o}}_{o\in\sigma(O)}$ on $\mathfrak{s}$ at $t_{1}>t_0$. Then there is a unitary operator $U$ that $\mathfrak{x}$ can use to establish $_\mathfrak{s}\!\ev{\pi_{\tilde{o}}^{t_{1}}}{v}_\mathfrak{s}={}_\mathfrak{s}\!\ev{U^{\dagger}\pi_{\tilde{o}}U}{v}_\mathfrak{s}$. \end{itemize}

Note the stipulation $\mathfrak{s}\neq\mathfrak{x}$. If this was absent, a contradiction would follow trivially, as Amanda could model her own state in contradictory ways.

Suppose now that Amanda is interested Gustavo's lab and has measured value 1 at $t_1$ ($M_\mathfrak{a}^{t_1} = 1$) and prepared system $\mathfrak{l}$ in $\ket*{+}_{\mathfrak{l}}$ at $t'>t_1$ accordingly. Using \textsf{U}, she can establish the joint state $\ket*{\text{fail}}_\mathfrak{lg}$ at $t_2$. Then, using her background knowledge of the experimental setup as well as the fact that $_\mathfrak{lg}\!\bra*{\text{fail}}\pi_{\text{fail}}\ket*{\text{fail}}_\mathfrak{lg}=1$, she can infer from \textsf{Q}:
\begin{itemize}
\item[$\alpha$]`I am certain that $M_\mathfrak{d}^{t_4}=$ fail.'
\end{itemize} 
Assume also that Gustavo measures $M_\mathfrak{g}^{t_2}=1$. Using \textsf{Q}, Gustavo may reason as follows: had the state of the qbit he received been $\ket*{0}_\mathfrak{l}$, he could have been certain that $M_\mathfrak{g}^{t_1}=0$, as $_\mathfrak{l}\!\ev{\pi_{0}}{0}_\mathfrak{l}=1$. So the state must have been $\ket*{+}_{\mathfrak{l}}$. Gustavo knows that Amanda has prepared the state accordingly, so he can reconstruct her reasoning by appeal to QT to conclude:
\begin{itemize}
\item[$\tilde{\gamma}$]`I am certain that $\mathfrak{a}$ is certain that $M_\mathfrak{d}^{t_4}=$ fail.'
\end{itemize} 
Or, using \textsf{C}: 
\begin{itemize}
\item[$\gamma$]`I am certain that $M_\mathfrak{d}^{t_4}=$ fail.'
\end{itemize}
Consider now the operator 
\begin{equation*}
\Pi_{\neg(\text{ok}_\mathfrak{c}\wedge 0_\mathfrak{g})}:=\mathbb{I}_4-\Pi_{\text{ok}}^{t_3}\Pi_{0}^{t_2},
\end{equation*} 
where
\begin{equation}\label{eq:PiokPi0}
\Pi_{0}^{t_2}=U^\dagger(\mathbb{I}_2\otimes\pi_{0}\otimes\mathbb{I})U, \ \  \Pi_{\text{ok}}^{t_3}=\bar{U}^\dagger(\pi_\text{ok}\otimes\mathbb{I}_2)\bar{U}.
\end{equation}
Here, $U=U_{t'}U_{t_1}$ evolves operators until $t= 2$ and $\bar{U}=U_{t_2} U$ until $t=3$, with the $U_{i}$ as given in the \hyperref[sec:append]{Appendix}, \hyperref[thm:OkZero]{A 1}. As shown there, this yields
\begin{equation}
{}_{\mathfrak{ralg}}\!\bra{\text{init}}\Pi_{\neg(\text{ok}_\mathfrak{c}\wedge 0_\mathfrak{g})}\ket*{\text{init}}_\mathfrak{ralg}=1 -{}_{\mathfrak{ralg}}\!\bra{\text{init}}\Pi_{\text{ok}}^{t_3}\Pi_{0}^{t_2}\ket*{\text{init}}_\mathfrak{ralg}=1-0=1,
\end{equation} 
with $\ket*{\text{init}}_\mathfrak{ralg}:= \ket*{\text{init}}_\mathfrak{r}\bigotimes_{j=\mathfrak{a},\mathfrak{l},\mathfrak{g}}\ket*{0}_{j}$.

Using \textsf{Q} and \textsf{U}, Chris can hence be certain that $M_\mathfrak{c}^{t_3}\neq$ ok $\vee M_{\mathfrak{g}}^{t_2}\neq 0$. However,  $M_{\mathfrak{c}}^{t_3}=$ ok, so under the assumption that $M_{\mathfrak{g}}^{t_2}=0$, Chris could be certain that $\neg(M_\mathfrak{c}^{t_3}\neq$ ok $\vee M_{\mathfrak{g}}^{t_2}\neq 0)$ which contradicts \textsf{S}. Hence Chris can infer $M_{\mathfrak{g}}^{t_2}=1$. Building on this, he can reconstruct Gustavo's reasoning to arrive at: 
\begin{itemize}
\item[$\chi$]`I am certain that $M_\mathfrak{d}^{t_4}=$ fail.'
\end{itemize}
Finally, Chris tells David his result. Following through Chris' chain of reasoning, David can infer that Gustavo must have received $\ket*{+}_\mathfrak{l}$ and Amanda must have measured $1$ at $t_1$. From \eqref{eq:David} we read off that David can then also infer:
\begin{itemize}
\item[$\delta$]`I am certain that $M_\mathfrak{d}^{t_4}=$ fail.''
\end{itemize}
However, indeed $M_\mathfrak{d}^{t_4}=$ ok, and David is certain of that. Hence, by \textsf{S}, David can and cannot be certain that $M_\mathfrak{d}^{t_4}=$ ok. Contradiction.

\section{Probabilistic certainty and systems of epistemic and doxastic logic}\label{sec:S5}
The above proof had a crucial reliance on notions such as `is certain' or `knows', and the exact lines of reasoning by appeal to these notions are, in some instances, not perfectly clear. The rules for reasoning about such notions are the subject matter of epistemic logic, as has been noted by Nurgalieva and del Rio, leading them to claim the ``[i]nadequacy of modal logic in quantum settings''.

This is a bold claim. However, it seems plausible that the above argument does constitute a puzzle when reconstructed within a suitable modal system. What, however, is the puzzle \emph{exactly} and what are possible solutions?

Let us first recall some fundamentals. Let  $\mathscr{L}$ a propositional language over a set of elementary propositions $p,q,\ldots$, subject to the usual recursive rules for $\vee, \wedge, \neg$. `$\supset$' denotes material implication and is defined as $\neg\varphi\vee\psi$; `$\equiv$' denotes material equivalence and is defined as $(\varphi\supset\psi)\wedge(\psi\supset\varphi)$. Adding for each agent $\mathfrak{x}$ from some set $A$ a modal operator $\square_\mathfrak{x}$ and defining $\Diamond_\mathfrak{x}\varphi\equiv_{df} \neg \square_\mathfrak{x} \neg\varphi$ for each $\mathfrak{x}$, one obtains an enriched language $\mathscr{L}_{\square}$ in which for every well-formed formula $\varphi$ from $\mathscr{L}$ $\square_\mathfrak{x}\varphi$ is well-formed as well. 

A \emph{frame} is an ordered pair $\mathcal{F}_A=\ev{W, R}$, where $R:A\rightarrow \mathcal{P}(W^{2})$ yields a binary relation $R_\mathfrak{x}\subseteq  W\times W$ over $W$ for every $\mathfrak{x}\in A$. Intuitively, $w R_\mathfrak{x}w'$ specifies the worlds $w'$ that are in some sense epistemically possible for $\mathfrak{x}$ w.r.t.\ $w$, so that the operators may ``unpack the account of knowledge encoded in the frame by $R$.'' \citep[974]{williamson2014}. So $w R_\mathfrak{x}w'$ means that $\mathfrak{x}$ holds it possible that she is in $w'$ if she is in $w$. The converse might not be true, i.e., $R_\mathfrak{x}$ need not be assumed symmetric. 

A \emph{model} is an ordered pair $\mathcal{M}_{P}=\ev{\mathcal{F}_A, V}$ defined some set $P$ of elementary propositions, where $V: P \rightarrow \mathcal{P}(W)$ is a \emph{valuation function} that associates with every proposition $p\in P$ the worlds $w\in W$ in which it is true. (We will suppress indices $A$ and $P$ for simplicity below.)

Finally, a \emph{pointed model} may be defined as a pair $\hat{\mathcal{M}}=\ev{\mathcal{M},\hat{w}}$, where $\hat{w}\in W$. In epistemic logic(s), truth of formulas $\xi$ may then be defined w.r.t.\ a pointed model $\hat{\mathcal{M}}$ as follows: 
\begin{itemize}
\item[]$\hat{\mathcal{M}}\models p$ iff $\hat{w}\in V(p)$,
\item[]$\hat{\mathcal{M}}\models \varphi\wedge\psi$ iff $\hat{\mathcal{M}}\models  \varphi$ and $\hat{\mathcal{M}}\models  \psi$,  
\item[]$\hat{\mathcal{M}}\models \neg\varphi$ iff $\hat{\mathcal{M}}\not\models  \varphi$,  
\item[]$\hat{\mathcal{M}}\models\square_\mathfrak{x} \varphi$ iff for all $w$ s.t.\ $\hat{w}R_\mathfrak{x}w$, $\ev{\mathcal{M}, w}\models\varphi$.
\end{itemize}

We may think of $\hat{w}$ as the \emph{actual world}, w.r.t\ which the epistemic conditions of $\mathfrak{x}$ are evaluated. If some formula $\varphi$ holds in all pointed versions $\ev*{\mathcal{M}, w}$ of $\mathcal{M}$, we may write $\mathcal{M}\models \varphi$.

As is well known, there are different axiom systems which allow to assign different \emph{meanings} to the operators $\square_\mathfrak{x}$ and $\Diamond_\mathfrak{x}$. The weakest system is \textsf{\textbf{K}},\footnote{Actually one obtains a system \textsf{\textbf{K}}$_n$ for $|A|=n$, not identical to \textsf{\textbf{K}} with only one primitive modal operator for $n>1$. The same holds for all systems discussed below; but we generally omit the index and allow loose talk of  `system \textbf{\textsf{X}}' instead of \textbf{\textsf{X}}$_n$.} defined, w.r.t.\ all agents $\mathfrak{x}$ in some set $A$, by the following axiom schemata and transformation rules: 
\begin{itemize}[align=left, itemindent=2.5em]
\item[\textsf{PC}]If $\varphi$ is a propositional tautology, then $\varphi$ is a theorem.
\item[\textsf{K}]$\square_\mathfrak{x}(\varphi\supset\psi)\supset(\square_\mathfrak{x}\varphi\supset\square_\mathfrak{x}\psi$).
\item[\textsf{MP}]From $\varphi$ and $\varphi\supset\psi$, infer $\psi$.
\item[\textsf{N}]From $\varphi$ infer $\square_\mathfrak{x}\varphi$.
\end{itemize}
The propositional tautologies in \textsf{PC} may include modal operators, as long as they are valid in virtue of their propositional form (as in $\square_\mathfrak{x}p\vee\neg\square_\mathfrak{x}p$). Moreover, the `necessitation rule', \textsf{N}, must be treated with caution: it only be applies to theorems, not to premises or assumptions, for otherwise it produces contradictions and other undesirable results \citep[e.g.][29]{ditmarsch}.

While \textsf{\textbf{K}} is ``a minimal modal logic'', as it ``happens to capture the validities of the [...] class of \emph{all} Kripke models'' \citep[26; orig.\ emph.]{ditmarsch}, it is known to be problematic when $\square_\mathfrak{x}\varphi$ is interpreted as `$\mathfrak{x}$ knows that $\varphi$'. \textsf{N}, for instance, expresses logical omniscience: any propositional tautology, regardless of its complexity, is \emph{known} by all of the agents in the set $A$. Similarly, \textsf{K}, together with \textsf{MP}, delivers the \emph{closure} of knowledge under known implication; a principle long suspected to be at the heart of arguments for radical skepticism \citep[][]{nozick1981}.

Similar considerations apply to Bayesian approaches to \emph{probability} though: Applying the Kolmogorov axioms as a formal model of agents' beliefs requires logical omniscience. This has not impaired the popularity of Bayesianism in present-day philosophy. Typically, such considerations are dealt with in terms of \emph{idealization}, meaning that Bayesian probability models are understood as idealized (normative) models of perfectly rational, logically omniscient agents. The same is true about epistemic logic: It's use is to investigate ``\emph{idealised} notions of knowledge, that do not necessarily hold for human beings.'' \citep[ibid.; orig.\ emph.]{ditmarsch} 

The following three additional axioms are typically considered in epistemic logic:
\begin{itemize}[align=left, itemindent=2.5em]
\item[\textsf{T}]$\square_\mathfrak{x}\varphi\supset\varphi$.
\item[\textsf{4}]$\square_\mathfrak{x}\varphi\supset\square_\mathfrak{x}\square_\mathfrak{x}\varphi$.
\item[\textsf{5}]$\neg\square_\mathfrak{x}\varphi\supset\square_\mathfrak{x}\neg\square_\mathfrak{x}\varphi$.
\end{itemize}
\textsf{\textbf{K}}+\textsf{T}+\textsf{4}+\textsf{5} is the system \textsf{\textbf{S5}}. With appropriate interpretation of the operators $\square_\mathfrak{x}$, the axioms of \textsf{\textbf{S5}} are sometimes referred to as the `axioms of knowledge' \citep{ditmarsch, delrio}. It requires the accessibility relation to be an \emph{equivalence relation}, in the sense that any formula that is valid in a frame $\mathcal{F}$ wherein every $R_\mathfrak{x}$ is an equivalence relation is a theorem of \textsf{\textbf{S5}} \citep[][61]{hughes}. To recall, $R_\mathfrak{x}$ is an equivalence relation over $W(\ni w,w',w'')$ if it is reflexive: $wR_\mathfrak{x}w$; transitive: if $wR_\mathfrak{x}w'$ and $w'R_\mathfrak{x}w''$, then $wR_\mathfrak{x}w''$; and symmetric: $wR_\mathfrak{x}w'$ iff $w'R_\mathfrak{x}w$. 

Axiom \textsf{T} establishes knowledge as a \emph{factive} notion: whatever is known must be \emph{true}. This excludes, e.g., usage of `knowledge' wherein a set of background beliefs merely \emph{taken} to be established as true is meant \citep[][307]{halpern1991}. \textsf{T} therefore delineates the appropriate \emph{target} of epistemic logic as the (narrower) concept that epistemologists typically aim to explicate. 

Axioms \textsf{4} and \textsf{5} are less innocent: \textsf{4} is the KK-principle, that one knows what one knows, not respected by externalist theories of knowledge. \textsf{5} is a corresponding negative version that could be questioned on the same grounds. Hence \textsf{\textbf{S5}} deals with an idealized \emph{internalist} conception of knowledge.

These complications are avoided if one assumes what is ``usually considered to be [...] the weakest system of epistemic interest'' \citep[Sect.\ 1]{hendricks}, namely \textbf{\textsf{T}}, which is \textsf{\textbf{K}}+\textsf{T}. \textbf{\textsf{T}} is sound and complete w.r.t.\ the class of all reflexive frames \citep[cf.][28]{ditmarsch}, where a frame $\mathcal{F}$ is said to be reflexive if $R$ yields a reflexive accessibility relation for all $\mathfrak{x}\in A$. (We may equally call (pointed) models reflexive etc.\ below.)

We have seen that Frauchiger and Renner sometimes talk about agents' knowledge, sometimes about their being \emph{certain}. In the latter case, we may interpret probability-1 assertions, as they arise from the quantum calculus, as expressions of maximal \emph{credences}, rather than of knowledge. Indeed, there are several reasons for considering quantum probabilities, even those equal to unity, expressions of subjective credences rather than unequivocal indicators of objective events \citep[e.g.][]{CFS2007}. 

Now knowledge may be factive, but certainty is generally not. It is perfectly coherent to say: `I was absolutely certain that $\varphi$, but it was wrong nevertheless'. Yet Gustavo inferred that he could be certain of an event in case Amanda was certain of it, which is erroneous if Amanda is \emph{mistaken}. Hence, if we take Frauchiger and Renner's appeals to certainty seriously, the theorem might be blocked.

Therefore, let us consider also system \textsf{\textbf{KD45}}, usually regarded as an appropriate logic for \emph{belief} or \emph{certainty} (\cite[39]{ditmarsch}; \cite{halpern1991}), which results from \textsf{\textbf{S5}} when \textsf{T} is replaced by: 
\begin{itemize}[align=left, itemindent=2.5em]
\item[\textsf{D}]$\square_\mathfrak{x}\varphi\supset\Diamond_\mathfrak{x}\varphi$
\end{itemize}
In \textsf{\textbf{KD45}}, $\square_\mathfrak{x}\varphi$ is understood as `$\mathfrak{x}$ is certain that $\varphi$', and \textsf{D} essentially tells us that what is firmly believed by $\mathfrak{x}$ is not regarded impossible by her. Note also that, while each $R_\mathfrak{x}$ is required to be reflexive in \textsf{\textbf{T}}, the validities of \textsf{\textbf{KD45}} are captured by frames that are serial: for all $w\in W$ there is a $w'$ such that $wR_\mathfrak{x}w'$; transitive; and Euclidean: for all $w,w', w''\in W$, if $wR_\mathfrak{x}w'$ and $wR_\mathfrak{x}w''$ then $w'R_\mathfrak{x}w''$. Seriality hence means that something is conceivable for every agent from every world, transitivity that what is conceivable for an agent from one world will also be conceivable for her from worlds from which this world is conceivable, and Eucildeanness that two conceivable worlds are mutually conceivable.  

Locutions of knowledge and certainty were directly connected to probability-1 claims in the proof, so we also need to establish a link between the two. Indeed, such a link has been established by Halpern \cite[][]{halpern1991}, in terms of \emph{probability structures} $\mathcal{N}=\ev{W, V, p_\mathfrak{x}}_{\mathfrak{x}\in A}$, where each $p_\mathfrak{x}$ delivers a probability measure over an algebra $\mathcal{A}\subseteq\mathcal{P}(W)$ and therefore satisfies $p_\mathfrak{x}(\qty{w})\geq0$, for all $w\in W$, and $\sum_{w\in W}p_\mathfrak{x}(\qty{w})=1$.\footnote{This restricts the $p_\mathfrak{x}$ to countable $W$, which is fully sufficient for our case. Cf.\ \cite[314]{halpern1991} for remarks on generalizations to uncountable spaces.}

Furthermore, let $\mathcal{N}$ some probability structure, and $W_{\varphi}=\qty{w\in W| \ev{\mathcal{N}, w}\models\varphi }$. Then, setting for any pointed probability structure $\hat{\mathcal{N}}=\ev{\mathcal{N}, \hat{w}}$
\begin{equation}\label{eq:semHalp}
\hat{\mathcal{N}}\models\square_{\mathfrak{x}}\varphi\text{ iff }p_{\mathfrak{x}}(W_\varphi)=1\tag{$\dagger$},
\end{equation}
one may prove the following theorem:
\begin{theorem}[\cite{halpern1991}]\emph{\textsf{\textbf{KD45}}} is complete and sound w.r.t.\ the class $\mathscr{N}_0$ of all probability structures.
\end{theorem}
\eqref{eq:semHalp} establishes a semantic connection between probability statements and expressions of subjective certainty. The theorem then establishes \textsf{\textbf{KD45}} as the appropriate logic for reasoning about this probabilistic notion of subjective certainty. Note also that \eqref{eq:semHalp} is equivalent to the following:
\begin{equation}\label{eq:semHalp'}
\hat{\mathcal{N}}\models\square_{\mathfrak{x}}\varphi\text{ iff } \ev{\mathcal{N},w'}\models\varphi, \forall w'\in W \text{ s.t. }p_{\mathfrak{x}}(\qty{w'})>0 \tag{$\dagger$'},
\end{equation}
\begin{proof}

Direction `$\Rightarrow$': Assume $p_\mathfrak{x}(W_\varphi)=1$. By normalization, we have $p_{\mathfrak{x}}(W)=1$, so decomposing $W$ into the disjoint subsets $W_\varphi$ and $W\setminus W_\varphi$, we have $p_{\mathfrak{x}}(W \setminus W_\varphi)=0$. Now decompose $W \setminus W_\varphi$ into a disjoint union of singletons $\qty{w}\subseteq W \setminus W_\varphi$. Then by additivity, $\sum_{w\in W\setminus W_\varphi}p_{\mathfrak{x}}(\qty{w})=0$, so by positivity, $p_{\mathfrak{x}}(\qty{w})=0$ for each of these singletons. Hence $w\in W \setminus W_\varphi$ implies $p_{\mathfrak{x}}(\qty{w})=0$, which, by contraposition and the definition of $W_\varphi$, establishes the desired result.

Direction `$\Leftarrow$': For arbitrary $w'\in W$, $p_{\mathfrak{x}}(\qty{w'})>0$ implies $\ev{\mathcal{N},w'}\models\varphi$, whence by definition of $W_\varphi,  w'\in W_\varphi$. Hence for any $w''\in W\setminus W_\varphi$ we have $p_{\mathfrak{x}}(\qty{w''})=0$ by modus tollens. Since $W\setminus W_\varphi$ is the disjoint union of these singletons and $W$ the disjoint union of $W_\varphi$ and $W\setminus W_\varphi$, we have $1=p_{\mathfrak{x}}(W)=p_{\mathfrak{x}}(W_\varphi) + p_{\mathfrak{x}}(W\setminus W_\varphi)=p_{\mathfrak{x}}(W_\varphi)$. \hfill $\blacksquare$
\end{proof}

Moreover, given some probability structure $\mathcal{N}$, let, for any $\mathfrak{x}\in A$, $F_\mathfrak{x}^{\mathcal{N}}=\lbrace w\in W| \exists\varphi: \ev{\mathcal{N}, w}\models\neg\varphi \wedge\square_\mathfrak{x}\varphi \rbrace$ the set of $\mathfrak{x}$'s false beliefs. One can then immediately also prove the following: 
\begin{theorem}[\cite{halpern1991}]Given some $\mathfrak{x}\in A, \mathcal{N}\in\mathscr{N}_{0}$. Then $p_{\mathfrak{x}}(F_\mathfrak{x}^{\mathcal{N}})=0$.
\end{theorem}
\begin{proof}
Case  $F_{\mathfrak{x}}^{\mathcal{N}}=\varnothing$ is trivial. Otherwise, for arbitrary $w\in F_{\mathfrak{x}}^{\mathcal{N}}$, there is some $\varphi$ which by \eqref{eq:semHalp} has $p_\mathfrak{x}(W_\varphi)=1$. However, by definition of $F_{\mathfrak{x}}^{\mathcal{N}}$,  $\ev{\mathcal{N}, w}\not\models\varphi$ and hence $\qty{w}\subseteq W \setminus W_\varphi$. Then $p_\mathfrak{x}(\qty{w})\leq p_\mathfrak{x}(W \setminus W_\varphi)=1-p_\mathfrak{x}(W_\varphi)=0$. Therefore, by additivity, $p_{\mathfrak{x}}(F_{\mathfrak{x}}^{\mathcal{N}})=\sum_{w\in F_{\mathfrak{x}}^{\mathcal{N}}}p_{\mathfrak{x}}(\qty{w})=0$. \hfill $\blacksquare$
\end{proof}
The content of this theorem may be regarded as $\mathfrak{x}$ being certain that her firm beliefs are not false. Indeed, defining $\mathscr{N}_1\subset\mathscr{N}_0$ to be the class of probability structures that assign $p_\mathfrak{x}(\qty{w})>0$ over each $p_{\mathfrak{x}}$'s domain, one can show that \textsf{\textbf{S5}} is sound and complete w.r.t.\ $\mathscr{N}_1$ \citep[310]{halpern1991}. In other words: adding that $\mathfrak{x}$ is right in taking her beliefs to be correct, one ends up with the standard logic of (idealized, internalist) knowledge.  

However, probability assignments might vary not only from agent to agent but also from world to world. Hence, consider also Halpern's \emph{generalized probability structures}  $\mathcal{N}=\ev{W, V, p_\mathfrak{x}^{w}}_{w\in W, \mathfrak{x}\in A}$, where for any pair $\ev{w,\mathfrak{x}}\in W\times A$, $p_\mathfrak{x}^w$ delivers a probability measure over an algebra $\mathcal{A}\subseteq\mathcal{P}(W)$.

Moreover, let $\supp(p_\mathfrak{x})=\qty{\ev*{w,w'}\in W^{2}|p_\mathfrak{x}^{w}(\qty{w'})>0}$, and $\mathcal{M}_\mathcal{N}$ a Kripke model defined (over $A$) in such a way that for every $\mathfrak{x}\in A$, ${R}_\mathfrak{x}= \supp(p_\mathfrak{x})$. Given this identification, we finally note the following theorem: 
\begin{theorem}[\cite{halpern1991}]Any extension of \textsf{\textbf{K}} that includes either \textsf{D} or \textsf{T} is sound and complete w.r.t.\ some class of probability structures. 
\end{theorem}

\section{Frauchiger-Renner reconsidered}\label{sec:proof}
Building on \cite[275--80]{delrio}, but with our refined inventory from Sect.\ \ref{sec:S5}, we can now establish the following reconstruction of the Frauchiger-Renner theorem. First off, like Nurgalieva and del Rio \cite[280]{delrio}, we assume all agents in $A=\qty{\mathfrak{a}, \mathfrak{c}, \mathfrak{d}, \mathfrak{g}}$ to share a Kripke model $\mathcal{M}^{FR}=\ev{W, R, V}$. This requires an extended language which includes symbols for state assignments and measurement outcomes. The set $P$ of elementary propositions over which the model is defined then comprises all value claims of the form $M_\mathfrak{x}^t=y$ (short: $y^{t}_\mathfrak{x}$) considered in the \textsf{FR-Protocol}, as well as the state assignments at given times in some run of the experiment. 

To adequately represent the latter, we let $\ket*{x}^{t}_\mathfrak{s}$ stand for `$\mathfrak{s}$ is in state $\ket*{x}$ at $t$'. This is an abuse of notation, as we now let the (indexed) ket-symbol refer to states as well as proposition about states, depending on context. Moreover, to track the (standard) usage of quantum state assignments in the proof as given above, we impose the following condition:
\begin{equation}\label{eq:tensor}
\forall t: \ket*{v(t)}_\mathfrak{xy}=\ket*{a}_\mathfrak{x}\otimes\ket*{b}_\mathfrak{y}\text{ iff }\ket*{a}_\mathfrak{x}^t\wedge\ket*{b}_\mathfrak{y}^t.\tag{$\otimes/\wedge$}
\end{equation}
Note that on the left-hand side, the ket-symbol denotes states whereas on the right-hand side, it denotes propositions.

Since we model the \textsf{FR-Protocol}, the complex propositions at times $t\in\qty{1,\ldots,5}$ therein must be assumed true in all $w\in W$. For further reference, we formalize these as:\footnote{It would be tempting to introduce an \emph{action language} \citep[112 ff.]{ditmarsch} in order to track the epistemic \emph{dynamics} involved. But this would complicate the formalism further while adding few insights. The same holds for a first order predicate language, which would introduce the need to fix interpretations in semantic proofs.}
\begin{itemize}
\item[$\varphi_0$] $\bigwedge_{t\in[0,1)}\ket*{\text{init}}_\mathfrak{r}^{t} \bigwedge_{t\in[0,t_1)}\ket{0}_\mathfrak{a}^t\bigwedge_{t\in[0,t_1]}\ket*{0}_\mathfrak{l}^t\bigwedge_{t\in[0,t_2)}\ket{0}_\mathfrak{g}^t$
\item[$\varphi_1$] $(0_\mathfrak{a}^{t_1} \vee 1_\mathfrak{a}^{t_1}) \wedge (0_\mathfrak{a}^{t_1} \supset  (\kket{0}^{t_1}_{\mathfrak{a}_\mathfrak{r}}\bigwedge_{t\in(t_1,t_2]}\ket*{0}^{t}_\mathfrak{l})) \wedge (1_\mathfrak{a}^{t_1} \supset  (\kket{1}^{t_1}_{\mathfrak{a}_\mathfrak{r}}\bigwedge_{t\in(t_1,t_2]}\ket*{+}^t_\mathfrak{l}))$
\item[$\varphi_2$] $(0_\mathfrak{g}^{t_2}\vee 1_\mathfrak{g}^{t_2})\wedge (0_\mathfrak{g}^{t_2}\supset \kket{0}^{t_2}_{\mathfrak{g}_\mathfrak{l}})\wedge(1_\mathfrak{g}^{t_2}\supset\kket{1}^{t_2}_{\mathfrak{g}_\mathfrak{l}})$
\item[$\varphi_3$] $(\text{ok}_\mathfrak{c}^{t_3}\vee \text{fail}_\mathfrak{c}^{t_3})\wedge (\text{ok}_\mathfrak{c}^{t_3}\supset  \bigwedge_{t\in[t_3,5]}\kket{\text{ok}}^t_{\mathfrak{c}_\mathfrak{ra}})\wedge(\text{fail}_\mathfrak{c}^{t_3}\supset \bigwedge_{t\in[t_3,5]}\kket{\text{fail}}^t_{\mathfrak{c}_\mathfrak{ra}})$
\item[$\varphi_4$] $(\text{ok}_\mathfrak{d}^{t_4}\vee \text{fail}_\mathfrak{d}^{t_4})\wedge (\text{ok}_\mathfrak{d}^{t_4}\supset  \bigwedge_{t\in[t_4,5]}\kket{\text{ok}}^t_{\mathfrak{d}_\mathfrak{lg}})\wedge(\text{fail}_\mathfrak{d}^{t_4}\supset \bigwedge_{t\in[t_4,5]}\kket{\text{fail}}^t_{\mathfrak{d}_\mathfrak{lg}})$
\item[$\varphi_5$] $((\text{ok}_\mathfrak{c}^{t_3}\supset\kkket{\text{ok}}^5_{\mathfrak{d}_{\mathfrak{c}_\mathfrak{ra}}})\wedge (\text{fail}_\mathfrak{c}^{t_3}\supset\kkket{\text{fail}}^5_{\mathfrak{d}_{\mathfrak{c}_\mathfrak{ra}}}))\wedge ((\text{ok}_\mathfrak{d}^{t_4}\supset\kkket{\text{ok}}^5_{\mathfrak{c}_{\mathfrak{d}_\mathfrak{lg}}})\wedge (\text{fail}_\mathfrak{d}^{t_4}\supset\kkket{\text{fail}}^5_{\mathfrak{c}_{\mathfrak{d}_\mathfrak{lg}}}))$
\end{itemize}

Here we have introduced some new notation, namely $\kket{y}_{\mathfrak{x}_\mathfrak{s}}$ for a state of $\mathfrak{x}$ that unambiguously indicates that system $\mathfrak{s}$ under study by $\mathfrak{x}$ was in state $\ket*{y}_\mathfrak{s}$. The meaning of $\kkket{y}_{\mathfrak{x}_{\mathfrak{y}_\mathfrak{s}}}$ should be obvious. Emphatically, we do \emph{not} impose that this state of affairs must have come about by a(n actual) measurement of sorts. $\mathfrak{x}$ may convince herself that $\ket*{y}^t_\mathfrak{s}$ by reasoning or communication.

Moreover, each $\mathfrak{x}\in A$ must be capable of applying the appropriate unitaries in order to determine the relevant Heisenberg operators in \textsf{U}. Hence that the unitaries that model the evolution of the global sate are those of \hyperref[thm:OkZero]{A 1} in the \hyperref[sec:append]{Appendix} is another assumption that should be considered part of the \textsf{FR-Protocol}. The restriction $\mathfrak{s}\neq\mathfrak{x}$ in \textsf{U} implies that Amanda cannot use these operators. For her, we let $U_\mathfrak{a}=  (\pi_0\otimes\mathbb{I} + \pi_1\otimes\sigma_x)$, which is just the $\mathfrak{lg}$-part of $U_{t_2}$, postselected for $M^{t_1}_\mathfrak{a}=1$. We will refer to the assumption that these are the `appropriate' unitaries for the respective agents as $\upsilon$ below.

Since $\varphi_{0-5}\equiv_{df}\bigwedge_{i=0}^{5}\varphi_i$ and $\upsilon$ will thereby hold in all worlds $w'\in W$ accessible by any kind of (non-trivial) access relation $R_\mathfrak{x}$ for an agent $\mathfrak{x}\in A$ from any world $w\in W$, we immediately get that $\mathcal{M}^{FR}\models\square_\mathfrak{x}(\varphi_{0-5}\wedge\upsilon)$ for all $\mathfrak{x}$. Moreover, since the protocol has been agreed upon by all agents, we assume that $\mathcal{M}^{FR}\models\square_\mathfrak{y}\square_\mathfrak{x}(\varphi_{0-5}\wedge\upsilon)$, for arbitrary $\mathfrak{x},\mathfrak{y}$. $\varphi_{0-5}$ and $\upsilon$, in other words, constitute the (common) \emph{common knowledge} of the model $\mathcal{M}^{FR}$.

Let us now, given some model $\mathcal{M}=\ev{W, R, V}$, stipulate the following semantic connection between quantum probabilities and the operators $\square_\mathfrak{x}$, in order to recapture the essence of \textsf{Q} and \textsf{U}: Given a pointed model $\hat{\mathcal{M}}=\ev{\mathcal{M}, \hat{w}}$ with $\hat{w}\in V(\kket{v}^{t_0}_{\mathfrak{x}_\mathfrak{s}})$,
\begin{equation}\label{eq:Q}
\ev{\mathcal{M}, \hat{w}}\models\square_{\mathfrak{x}}\varphi_{t}\text{ iff } {}_\mathfrak{s}\!\ev{\pi_\varphi^{t}}{v}_\mathfrak{s} =1\tag{$\ast$},
\end{equation}
where $\kket{v}_{\mathfrak{x}_\mathfrak{s}}$ unambiguously indicates that $\mathfrak{s}$'s state was $\ket*{v}_\mathfrak{s}$ at $t_0$, $\pi_\varphi^{t}$ is a Heisenberg operator at $t\geq t_0$, and $\varphi_t$ concerns some value measured at $t$, represented by $\pi_\varphi^t$ in the quantum formalism. 

It is easy to see that \eqref{eq:Q} partitions $W^{2}$ as follows: Given some model $\mathcal{M}=\ev{W, R, V}$ over a set $A$ of agents, $R_\mathfrak{x}=\qty{\ev{w,w'}\in W^2 | wR_\mathfrak{x}w'}$, and we let $\bar{R}_\mathfrak{x}=\qty{\ev{w,w'}\in W^2 | \neg wR_\mathfrak{x}w'}$ for each $\mathfrak{x}\in A$. So $W^{2}={R}_\mathfrak{x}\cup\bar{{R}}_\mathfrak{x}$ and ${R}_\mathfrak{x}\cap\bar{{R}}_\mathfrak{x}=\varnothing$. Let also $Y=\lbrace \ev*{w,w'}\in W^2|\exists t, y: \ev{\mathcal{M}, w}\models \kket{y}^{t_0}_{\mathfrak{x}_\mathfrak{s}} \text{ and } \ev{\mathcal{M}, w'}\models \neg y_\mathfrak{x}^{t}\rbrace$. Then $\mathcal{M}$ satisfies \eqref{eq:Q}  only if $Y\subseteq\bar{{R}}_\mathfrak{x}$ and hence ${R}_\mathfrak{x}\cap Y=\varnothing$, for all $\mathfrak{x}\in A$. 

Clearly, there is a close connection between Halpern's generalized probability structures and \eqref{eq:Q}. However, \eqref{eq:Q} constrains the reference world $\hat{w}$ from which accessibilities are evaluated in such a way that $\hat{w}\in V(\kket{v}^{t_0}_{\mathfrak{x}_\mathfrak{s}})$. We may read this condition as stating that, relative to $\hat{w}$ in which $\mathfrak{x}$ assigns $\ket*{v}_\mathfrak{s}$ to $\mathfrak{s}$, $\mathfrak{x}$ gives measure 1 to the subset of $W$ in which $\varphi$ is true. Moreover, given the equivalence between \eqref{eq:semHalp} and \eqref{eq:semHalp'},\footnote{In addition, consider Thm.\ 5.1 in \cite[][312]{halpern1991}.} another way of looking at this condition is that, whatever the global quantum states $\ket*{\Psi}_{\hat{w}}, \ket*{\Psi'}_w$ for $\hat{w}$ and $w$, respectively, so long as $\ket*{\Psi}_{\hat{w}}$ is of the form $\ket*{\Psi}_{\hat{w}}=\ldots\otimes\kket{v}_{\mathfrak{x}_\mathfrak{s}}\otimes\ldots$, $\varphi$ holds true in any world $w$ such that $|{}_{\hat{w}}\!\braket*{\Psi}{\Psi'}_w|$ is non-zero.\footnote{Note, however, that this condition may never \emph{strictly} apply: interactions between systems are basically omnipresent, and they often entangle the systems involved. A global quantum state for a whole possible world would hence be unlikely to have the required product-state form. The condition should rather be considered in the asymptotic limit of free states for single systems.}

\begin{remark}[Nested box-formulas] \eqref{eq:Q} also allows nested assertions of the form $\square_\mathfrak{x}\square_\mathfrak{y}\varphi$: Assume that $\hat{w}\in V(\kkket{v}^{t_\mathfrak{x}}_{\mathfrak{x}_{\mathfrak{y}_\mathfrak{s}}})$ at some arbitrary time $t_\mathfrak{x}$, where $\kkket{v}^{t_\mathfrak{x}}_{\mathfrak{x}_{\mathfrak{y}_\mathfrak{s}}}$ unambiguously indicates that $\kket{v}^{t_\mathfrak{y}}_{\mathfrak{y}_\mathfrak{s}}$, with $\mathfrak{y}$ some agent, and this state in turn unambiguously indicates $\ket*{v}_\mathfrak{s}$ on some system $\mathfrak{s}$ under study by $\mathfrak{y}$ at $t_\mathfrak{y}$. Then given that ${}_{\mathfrak{y}_\mathfrak{s}}\!\langle\!\bra{v}\!|\pi_v^{t_\mathfrak{x}}\kket{v}_{\mathfrak{y}_\mathfrak{s}}=1$ (where $\pi_v^{t_\mathfrak{x}}=|\!\dyad*{v\rangle\!}{\!\langle v}\!|$), if ${}_\mathfrak{s}\!\ev*{\pi_\varphi^{t}}{v}_\mathfrak{s} =1$ for $t\geq t_\mathfrak{y}$, $\square_\mathfrak{x}\square_\mathfrak{y}\varphi_t$ holds in $\hat{w}$. \hfill $\blacklozenge$
\end{remark}

Upon imposing \eqref{eq:Q}, we can immediately see that \textsf{S} follows in \textsf{\textbf{T}} in the following way: 
\begin{lemma}\label{lm:1}Given \eqref{eq:Q}, \textsf{S} is a theorem of \textsf{\textbf{T}}
\end{lemma}
\begin{proof}Suppose that ${}_\mathfrak{s}\!\ev{\pi_o^t}{v}_\mathfrak{s} =1$. Then by \eqref{eq:Q}, a given pointed model $\hat{\mathcal{M}}= \ev{\mathcal{M}, \hat{w}}$ that satisfies $\hat{w}\in V(\kket{v}^{t_0}_{\mathfrak{x}_\mathfrak{s}})$ for $t_0\leq t$ will also satisfy $\hat{\mathcal{M}}\models\square_\mathfrak{x}O_t=o$ and vice versa. Hence, \textsf{S} is equivalent to assuming $\hat{\mathcal{M}}\models\neg(\square_\mathfrak{x}O_t=o \wedge \square_\mathfrak{x}\neg O_t= o)$ for any such model. Using \textsf{PC} and the definition of $\Diamond_\mathfrak{x}$, this is equivalent to $\hat{\mathcal{M}}\models\neg\square_\mathfrak{x}O_t=o \vee \Diamond_\mathfrak{x}O_t= o$, which in turn is equivalent to $\hat{\mathcal{M}}\models\square_\mathfrak{x}O_t=o \supset \Diamond_\mathfrak{x}O_t= o$. This formula is true in reflexive models, as it is an instance of axiom \textsf{D} which follows from \textsf{T}. (We can easily see this: assume that, for any $\varphi, \mathfrak{x}$, $\square_\mathfrak{x}\varphi$ holds. From \textsf{T}, we then get $\varphi$. Now assume also that $\neg\Diamond_\mathfrak{x}\varphi$, which is equivalent to $\square_\mathfrak{x}\neg\varphi$. Using \textsf{T} again, we get a contradiction $(\varphi \wedge \neg\varphi)$.)  \hfill $\blacksquare$
\end{proof}

The above proof of \textsf{S} from \eqref{eq:Q} would have been possible in a weaker axiom system that includes \textsf{D} instead of \textsf{T}. However, recovering \textsf{C} within \textsf{\textbf{T}} as follows shows what is really at stake:
\begin{lemma}Given \eqref{eq:Q}, \textsf{C} is a theorem of \textsf{\textbf{T}}\label{lm:2}
\end{lemma}
\begin{proof}By \eqref{eq:Q} and the above remark on nested box-formulas, proving \textsf{C} within \textsf{\textbf{T}} boils down to proving $\hat{\mathcal{M}}\models \square_\mathfrak{x}\square_\mathfrak{y}\varphi\supset\square_\mathfrak{x}\varphi$, for $\hat{\mathcal{M}}=\ev{\mathcal{M}, \hat{w}}$ reflexive and $\hat{w}\in V(\kkket{v}^t_{\mathfrak{x}_\mathfrak{y}})$, with $\kkket{v}_{\mathfrak{x}_\mathfrak{y}}$ and $t$ as required. Hence, assume that $\hat{\mathcal{M}}\models \square_\mathfrak{x}\square_\mathfrak{y}\varphi$ for some  such pointed model. Then it holds that for any $w$ s.t.\ $\hat{w}R_\mathfrak{x}w$, $\ev{\mathcal{M}, w}\models \square_\mathfrak{y}\varphi$. In turn this means that for any $w'$ s.t.\ $w R_\mathfrak{y}w'$ it holds that $\ev{\mathcal{M}, w'}\models\varphi$. However, by reflexivity of $R_\mathfrak{y}$, $\ev{\mathcal{M}, w}\models \varphi$. Since $w$ is otherwise arbitrary, $\hat{\mathcal{M}}\models\square_\mathfrak{x}\varphi$. \hfill $\blacksquare$\end{proof} 

With \eqref{eq:Q} in place, we could have also proven this syntactically by a sequential application of the factivity axiom, \textsf{T}, and necessitation in the final step. Accordingly,  we have only used reflexivity in the proof, so we can see immediately that system \textsf{\textbf{T}} suffices. This underscores Nurgalieva and del Rio's \cite[278]{delrio} observation that ``[t]he introspection axioms'', \textsf{4} and \textsf{5}, ``are not directly applied in the Frauchiger-Renner setting[...].'' But, they also think that they ``resemble assumption \textsf{S}[...].'' (ibid.) As we have just seen, this coveys an unnecessarily strong status on \textsf{S}: it does not require an invocation of symmetry or transitivity. 

Still, in case an agent knows the outcome of some measurement, we may associate a quantum state to her that indicates this unambiguously. This is weaker than (positive) introspection (\textsf{4}), since it does not concern higher order knowledge. However, it may give rise to further knowledge via condition \eqref{eq:Q}. 

Now while \textsf{\textbf{T}} is usually considered `minimal' as a meaningful epistemic logic, it is nevertheless \emph{epistemic} in the sense of establishing a \emph{factive} notion of knowledge. Hence, when Frauchiger and Renner talk about certainty and impose \textsf{C}, they really must have \emph{knowledge} in mind.

The Frauchiger-Renner theorem may now be restated as follows:\footnote{Note that, like Nurgalieva and del Rio, we will allow some benign mixing of syntax and semantics in abbreviating more or less obvious semantic proof sequences by invocations of the corresponding axioms. We will generally also allow some amount of sloppiness when it comes to the exact involvement of propositional logic.}
\begin{theorem}[FR; \cite{frauchiger}]Given conditions \eqref{eq:Q}, \textsf{U}, as well as a reflexive Kripke model $\mathcal{M}^{FR}=\ev{W, R, V}$ over $A=\qty{\mathfrak{a}, \mathfrak{c}, \mathfrak{d}, \mathfrak{g}}$ and with $P$ containing all value statements and state assignments in the \textsf{FR-Protocol} such that $\varphi_{0-5}$ and $\upsilon$ are valid in all $w\in W$. Then any $\hat{w}\in W$ such that $\hat{w}\in V(1^{t_1}_\mathfrak{a})\cap V(1^{t_2}_\mathfrak{g})\cap V(\text{ok}^{t_3}_\mathfrak{c})\cap V(\text{ok}^{t_4}_\mathfrak{d})$ satisfies $\ev{\mathcal{M}^{FR}, \hat{w}}\models \perp$.
\end{theorem}
\begin{proof}
Let $\hat{\mathcal{M}}^{FR}=\ev{\mathcal{M}^{FR}, \hat{w}}$ and $\hat{w}\in V(1^{t_1}_\mathfrak{a})\cap V(1^{t_2}_\mathfrak{g})\cap V(\text{ok}^{t_3}_\mathfrak{c})\cap V(\text{ok}^{t_4}_\mathfrak{d})$ as required. 

(i) Based on $\hat{w}\in V(1^{t_1}_\mathfrak{a})$, we get $\hat{\mathcal{M}}^{FR}\models\kket{1}^{t_1}_{\mathfrak{a}_\mathfrak{r}}$ from \textsf{MP} on $\varphi_1$. Consider the operator $\pi_{1_\mathfrak{a}}^{t_1}=\pi_1$. Since ${}_\mathfrak{r}\!\ev*{\pi_{1_\mathfrak{a}}^{t_1}}{1}_\mathfrak{r}=1$, we may infer from \eqref{eq:Q} that $\hat{\mathcal{M}}^{FR}\models \square_\mathfrak{a}1_{\mathfrak{a}}^{t_1}$. Since any of the conjuncts from $\varphi_1$ is valid in all $w\in W$, this, together with \textsf{K}, establishes that $\hat{\mathcal{M}}^{FR}\models\square_\mathfrak{a}\ket*{+}^{t'}_\mathfrak{l}$, for $t'$ right after $t_1$, when the preparation is finished. Similarly, any of the conjuncts in $\varphi_0$ holds in all $w\in W$, so that $\hat{\mathcal{M}}^{FR}\models\square_\mathfrak{a}\ket*{0}^{t'}_\mathfrak{g}$. Hence, in any $w$ s.t.\ $\hat{w}R_\mathfrak{a}w$, it holds that $\ev{\mathcal{M}^{FR}, w}\models \ket*{+}^{t'}_\mathfrak{l}\wedge\ket*{0}^{t'}_\mathfrak{g}$, which by \eqref{eq:tensor} means that $\ev{\mathcal{M}^{FR}, w}\models \ket*{+}^{t'}_\mathfrak{l}\ket*{0}^{t'}_\mathfrak{g}$, and therefore $\hat{\mathcal{M}}^{FR}\models\square_\mathfrak{a}\ket*{+}^{t'}_\mathfrak{l}\ket*{0}^{t'}_\mathfrak{g}$. Based on this knowledge, we can assume $\mathfrak{a}$ to be in a corresponding state $\kket{+,0}_{\mathfrak{a}_\mathfrak{lg}}^{t'}$. By \textsf{U} and $\upsilon$, $\mathfrak{a}$ can now appeal to $U_\mathfrak{a}$ to obtain ${}_\mathfrak{l}\!\bra*{+} {}_\mathfrak{g}\!\bra*{0}\pi_{\text{fail}_\mathfrak{d}}^{t_4}\ket*{+}_\mathfrak{l}\ket*{0}_\mathfrak{g}=1$, as is shown in the \hyperref[sec:append]{Appendix}, \hyperref[thm:Amand]{A 4}, where $\pi_{\text{fail}_\mathfrak{d}}^{t_4}=U^{\dagger}_\mathfrak{a}\pi_\text{fail}U_\mathfrak{a}$. Therefore, by \eqref{eq:Q}, $\hat{\mathcal{M}}^{FR}\models\square_\mathfrak{a}\text{fail}_\mathfrak{d}^{t_4}$.

(ii) By assumption $\hat{w} \in  V(1^{t_2}_\mathfrak{g})$. So by $\varphi_2$, $\hat{w}\in V(\kket{1}^{t_2}_{\mathfrak{g}_\mathfrak{l}})$. Since ${}_\mathfrak{l}\!\ev*{\pi_{\neg 0_\mathfrak{l}}^{t_2}}{1}_\mathfrak{l}=1$, where $\pi_{\neg 0_\mathfrak{l}}^{t_2}:=\mathbb{I}-\pi_0$, \eqref{eq:Q} now implies that $\hat{\mathcal{M}}^{FR}\models\square_\mathfrak{g}\neg\ket*{0}^{t_2}_\mathfrak{l}$. It is easy to prove that $\varphi_1$ \textsf{PC}-implies $\neg\ket*{0}^{t_2}_\mathfrak{l}\supset1_\mathfrak{a}^{t_1}$, and therefore, by \textsf{MP}, ($\mathfrak{g}$'s knowledge that) $\varphi_1$, and \textsf{K}, $\hat{\mathcal{M}}^{FR}\models\square_\mathfrak{g}1_\mathfrak{a}^{t_1}$. It similarly follows that $\hat{\mathcal{M}}^{FR}\models\square_\mathfrak{g}\kket{1}^{t_1}_{\mathfrak{a}_\mathfrak{r}}$. Then in any $w\in W$ s.t.\ $\hat{w}R_\mathfrak{g}w$, $\ev*{\mathcal{M}^{FR}, w}\models\kket{1}^{t_1}_{\mathfrak{a}_\mathfrak{r}}$. As shown in (i), this suffices to establish that $\ev*{\mathcal{M}^{FR}, w}\models\square_\mathfrak{a}1^{t_1}_{\mathfrak{a}}$. Hence, $\hat{\mathcal{M}}^{FR}\models\square_\mathfrak{g}\square_\mathfrak{a}1_\mathfrak{a}^{t_1}$. 

Assume now that $\hat{\mathcal{M}}^{FR}\models\Diamond_\mathfrak{g}\neg\square_\mathfrak{a}\text{fail}_\mathfrak{d}^{t_4}$. Then there is a $w$ s.t.\ $\hat{w}R_\mathfrak{g}w$ and $\ev*{\mathcal{M}^{FR}, w}\models\neg\square_\mathfrak{a}\text{fail}_\mathfrak{d}^{t_4}$. However,  $\hat{\mathcal{M}}^{FR}\models\square_\mathfrak{g}\square_\mathfrak{a}(\varphi_{0-5}\wedge\upsilon)$ and, as we have just seen, $\hat{\mathcal{M}}^{FR}\models\square_\mathfrak{g}\square_\mathfrak{a}1_\mathfrak{a}^{t_1}$. So for all $w'$ s.t.\ $\hat{w}R_\mathfrak{g}w'$, it holds that $\ev*{\mathcal{M}^{FR}, w'}\models\square_\mathfrak{a}1_\mathfrak{a}^{t_1}\wedge\square_\mathfrak{a}(\varphi_{0-5}\wedge\upsilon)$, in particular also for $w'=w$. However, as show in (i), this suffices to establish $\ev*{\mathcal{M}^{FR}, w}\models\square_\mathfrak{a}\text{fail}_\mathfrak{d}^{t_4}$. Contradiction. Hence $\hat{\mathcal{M}}^{FR}\models\square_\mathfrak{g}\square_\mathfrak{a}\text{fail}_\mathfrak{d}^{t_4}$, and Lm.\ \ref{lm:2} gives us $\hat{\mathcal{M}}^{FR}\models\square_\mathfrak{g}\text{fail}_\mathfrak{d}^{t_4}$.

(iii) From the fact that $\kket{\text{init}}^0_{\mathfrak{c}_\mathfrak{ralg}}$ and ${}_{\mathfrak{ralg}}\!\bra{\text{init}}\Pi_{\neg(\text{ok}_\mathfrak{c}\wedge 0_\mathfrak{g})}\ket*{\text{init}}_\mathfrak{ralg}=1$ (cf.\ \hyperref[sec:append]{Appendix}, \hyperref[thm:OkZero]{A 1}), with $\Pi_{\neg(\text{ok}_\mathfrak{c}\wedge 0_\mathfrak{g})}$ defined as in \eqref{eq:PiokPi0}, we immediately get that $\hat{\mathcal{M}}^{FR}\models\square_\mathfrak{c}(\text{ok}_\mathfrak{c}^{t_3}\supset \neg 0_{\mathfrak{g}}^{t_2})$. By $\hat{w}\in V(\text{ok}^{t_3}_\mathfrak{c})$ and $\varphi_3$, we get $\hat{w}\in V(\kket{\text{ok}}_{\mathfrak{c}_\mathfrak{ra}}^{t_3})$. $_\mathfrak{ra}\!\ev*{\pi_{\text{ok}_\mathfrak{c}}^{t_3}}{\text{ok}}_\mathfrak{ra}=1$, where $\pi_{\text{ok}_\mathfrak{c}}^{t_3}=\pi_{\text{ok}}$, gives us $\hat{\mathcal{M}}^{FR}\models\square_\mathfrak{c}\text{ok}_\mathfrak{c}^{t_3}$. So by \textsf{K}, $\hat{\mathcal{M}}^{FR}\models\square_\mathfrak{c}\neg 0_{\mathfrak{g}}^{t_2}$. In turn, \textsf{PC}, ($\mathfrak{c}$'s knowledge that) $\varphi_2$, and \textsf{K} establish that $\hat{\mathcal{M}}^{FR}\models\square_\mathfrak{c}1_{\mathfrak{g}}^{t_2}$. 

From ($\mathfrak{c}$'s knowledge of) $\varphi_2$, and \textsf{K}, we immediately also get $\hat{\mathcal{M}}^{FR}\models\square_\mathfrak{c}\kket{1}^{t_2}_{\mathfrak{g}_\mathfrak{l}}$. Thus, for any $w\in W$ s.t.\ $\hat{w}R_\mathfrak{c}w$, it holds that $w\in V(\kket{1}^{t_2}_{\mathfrak{g}_\mathfrak{l}})$. Hence, as shown in (ii), it follows that $\ev*{\mathcal{M}^{FR}, w}\models\square_\mathfrak{g}\neg\ket*{0}^{t_2}_\mathfrak{l}$ in any such $w$, and as is shown there as well, $\ev*{\mathcal{M}^{FR}, w'}\models\square_\mathfrak{a}1_\mathfrak{a}^{t_1}$ for any $w'\in W$ s.t.\ $wR_\mathfrak{g} w'$. So $\ev*{\mathcal{M}^{FR}, w}\models\square_\mathfrak{g}\square_\mathfrak{a}1_\mathfrak{a}^{t_1}$. Since $w$ is otherwise arbitrary, we immediately get $\hat{\mathcal{M}}^{FR}\models\square_\mathfrak{c}\square_\mathfrak{g}\square_\mathfrak{a}1_\mathfrak{a}^{t_1}$. By Lm.\ \ref{lm:2}, we can reduce this to $\hat{\mathcal{M}}^{FR}\models\square_\mathfrak{c}\square_\mathfrak{a}1_\mathfrak{a}^{t_1}$, and the proof that $\hat{\mathcal{M}}^{FR}\models\square_\mathfrak{c}\text{fail}_\mathfrak{d}^{t_4}$ is immediate. 

(iv) Finally, from the common common knowledge of $\varphi_0$, we get that David will not only be in state $\kket{\text{init}}_{\mathfrak{d}_\mathfrak{ralg}}$ at $t=0$, but also in $\kkket{\text{init}}_{\mathfrak{d}_{\mathfrak{x}_\mathfrak{ralg}}}$ w.r.t.\ any $\mathfrak{x}\in A$. So by the remark on nested box-formulas and (iii), we get that $\hat{\mathcal{M}}^{FR}\models\square_\mathfrak{d}\square_\mathfrak{c}(\text{ok}_\mathfrak{c}^{t_3}\supset\neg 0_{\mathfrak{g}}^{t_2})$.  Moreover, from $\varphi_5$ we get that $\kkket{\text{ok}}^5_{\mathfrak{d}_{\mathfrak{c}_\mathfrak{ra}}}$, and so in basically the same way, we obtain $\hat{\mathcal{M}}^{FR}\models\square_\mathfrak{d}\square_\mathfrak{c}\text{ok}_\mathfrak{c}^{t_3}$. Thus, by \textsf{K}, in any $w$ s.t.\ $\hat{w}R_\mathfrak{d}w$, it holds that $\ev*{\mathcal{M}^{FR}, w}\models\square_\mathfrak{c}\neg 0_{\mathfrak{g}}^{t_2}$. Hence $\hat{\mathcal{M}}^{FR}\models\square_\mathfrak{d}\square_\mathfrak{c}\neg 0_{\mathfrak{g}}^{t_2}$ and by Lm.\ \ref{lm:2} $\hat{\mathcal{M}}^{FR}\models\square_\mathfrak{d}\neg 0_{\mathfrak{g}}^{t_2}$. The proof that $\hat{\mathcal{M}}^{FR}\models\square_\mathfrak{d}\text{fail}_{\mathfrak{d}}^{t_4}$ is now immediate.  

However, $\hat{w}\in V(\text{ok}^{t_4}_\mathfrak{d})$ and so by $\varphi_4$, $\hat{w}\in V(\kket{\text{ok}}^{t_4}_{\mathfrak{d}_\mathfrak{lg}})$. Considering that $_\mathfrak{lg}\!\ev*{\pi_{\neg\text{fail}_\mathfrak{d}}^{t_4}}{\text{ok}}_\mathfrak{lg}=1$, where $\pi_{\neg\text{fail}_\mathfrak{d}}^{t_4}=\mathbb{I}-\pi_{\text{fail}}$, we get $\hat{\mathcal{M}}^{FR}\models\square_\mathfrak{d}\neg\text{fail}_\mathfrak{d}^{t_4}$. From Lm.\ \ref{lm:1}, we get a contradiction. \hfill $\blacksquare$
\end{proof}

The proof of Thm.\ FR has the form of a reductio, leading from the assumption of a world $\hat{w}$ to a contradiction. Moreover, it is tempting to blame \emph{reflexivity}, as it figured crucially in the proof of Lm.\ \ref{lm:2}, and we can easily provide a counterexample to Lm.\ \ref{lm:2} in an serial, transitive, Euclidean frame. Consider the model presented in Fig.\ \ref{fig:no2}: Every agent `can see' some world from any world (seriality), if two (not necessarily distinct) worlds are seen by an agent from some world, they see each other (Euclideanness), and worlds that can be seen from seen worlds are also seen (transitivity). Yet from every world that $\mathfrak{x}$ sees from $w_{\neg\varphi}$, $\mathfrak{y}$ only sees worlds in which $\varphi$ is true ($\square_\mathfrak{x}\square_\mathfrak{y}\varphi$). But in $w_{\neg\varphi}$, $\varphi$ is false, and $\mathfrak{x}$ can see this ($\Diamond_\mathfrak{x}\neg\varphi$). 
\begin{figure}[htb!]
\begin{center}
\includegraphics[scale=0.3]{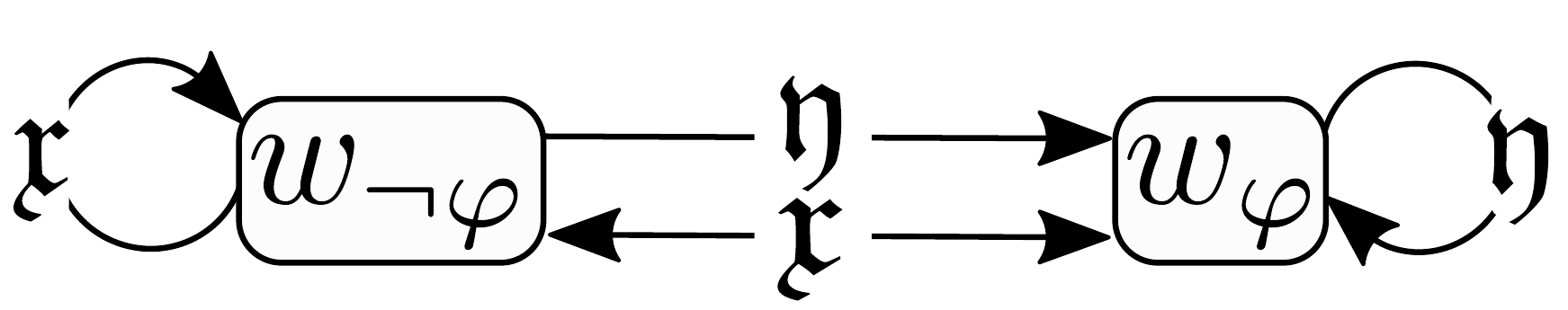}
\captionsetup{justification=centering, labelfont=bf, width=0.9\textwidth}
\caption{Serial, transitive, Euclidean counterexample to Lm.\ \ref{lm:2}.}
\label{fig:no2}
		\end{center}
\end{figure}

However, consider the following, stronger theorem: 

\begin{theorem}[FR*]Given conditions \eqref{eq:Q}, \textsf{U}, as well as a serial Kripke model $\mathcal{M}^{FR}=\ev{W, R, V}$ over $A=\qty{\mathfrak{a}, \mathfrak{c}, \mathfrak{d}, \mathfrak{g}}$ and with $P$ containing all value statements and state assignments in the \textsf{FR-Protocol}, such that all $w\in W$ satisfy $\varphi_{0-5}$ and $\upsilon$. Then there is a $\hat{w}\in W$ such that $\hat{w}\in V(1^{t_1}_\mathfrak{a})\cap V(1^{t_2}_\mathfrak{g})\cap V(\text{ok}^{t_3}_\mathfrak{c})\cap V(\text{ok}^{t_4}_\mathfrak{d})$ and $\ev*{\mathcal{M}^{FR}, \hat{w}}\models\perp$.
\end{theorem}
\begin{proof}
(I) Let $\varphi_{FR}\equiv 1^{t_1}_\mathfrak{a}\wedge 1^{t_2}_\mathfrak{g}\wedge \text{ok}^{t_3}_\mathfrak{c}\wedge \text{ok}^{t_4}_\mathfrak{d}$. By the agreed upon $\varphi_0$ and \eqref{eq:tensor}, $\ket*{\text{init}}^{0}_\mathfrak{ralg}=\ket*{\text{init}}_\mathfrak{r}\bigotimes_{j=\mathfrak{a},\mathfrak{l},\mathfrak{g}}\ket*{0}_j$. This is known by Chris and David, so $\kket{\text{init}}_{\mathfrak{c}/\mathfrak{d}_\mathfrak{ralg}}^0$ holds in any $w\in W$. Consider now the Heisenberg operator $\Pi_{\neg\varphi_{FR}}=\mathbb{I}_4 - \Pi_\text{ok}^{t_4}\Pi_\text{ok}^{t_3}\Pi_1^{t_2}\Pi_1^{t_1}$, defined in terms of projectors on $\mathcal{H}_\mathfrak{ralg}$ and evolved under unitaries $U_{i}$ from the \hyperref[sec:append]{Appendix}, \hyperref[thm:OkZero]{A 1}. As shown in the \hyperref[sec:append]{Appendix}, \hyperref[thm:JointOcc]{A 2}, ${}_\mathfrak{ralg}\!\ev{\Pi_{\neg\varphi_{FR}}}{\text{init}}_\mathfrak{ralg}=1-{}_\mathfrak{ralg}\!\ev{\Pi_\text{ok}^{t_4}\Pi_\text{ok}^{t_3}\Pi_1^{t_2}\Pi_1^{t_1}}{\text{init}}_\mathfrak{ralg} = 11/12\neq 1$. From \eqref{eq:Q} it now follows that $\mathcal{M}^{FR}\models\neg\square_\mathfrak{x}\neg\psi_{FR}$ for $\mathfrak{x}\in \qty{\mathfrak{c},\mathfrak{d}}$, who can avail themselves of the relevant unitaries on account of $\upsilon$. Hence, there is a $\hat{w}\in W\cap V(1^{t_1}_\mathfrak{a})\cap V(1^{t_2}_\mathfrak{g})\cap V(\text{ok}^{t_3}_\mathfrak{c})\cap V(\text{ok}^{t_4}_\mathfrak{d})$, accessible to $\mathfrak{c}, \mathfrak{d}$ from any $w\in W$.

(II) As in step (i) of Thm.\ FR, we get $\hat{w}\in V(\kket{+,0}_{\mathfrak{a}_\mathfrak{lg}}^{t'})$, and as shown in the \hyperref[sec:append]{Appendix}, \hyperref[thm:Amand]{A 4}, ${}_\mathfrak{l}\!\bra*{+} {}_\mathfrak{g}\!\bra*{0}\pi_{\neg\text{ok}_\mathfrak{d}}^{t_4}\ket*{+}_\mathfrak{l}\ket*{0}_\mathfrak{g}=1$, where $\pi_{\neg\text{ok}_\mathfrak{d}}^{t_4}=U^{\dagger}_\mathfrak{a}(\mathbb{I}_{2}-\pi_\text{ok})U_\mathfrak{a}$. Then $\hat{\mathcal{M}}^{FR}\models\square_\mathfrak{a}\neg\text{ok}_\mathfrak{d}^{t_4}$. Assume now that $W=\qty{\hat{w}}$. Then the seriality of the accessibilities collapses into reflexivity, and so $\hat{\mathcal{M}}^{FR}\models\neg\text{ok}_\mathfrak{d}^{t_4}$. Contradiction. Hence, there must be $w\neq \hat{w}$ in $W$, s.t.\ $\hat{w}R_\mathfrak{a} w$ and $\ev*{\mathcal{M}^{FR}, w}\models\neg\text{ok}_\mathfrak{d}^{t_4}$. 

However, on account of $\varphi_4$, we then get $w\in V(\text{fail}_\mathfrak{d}^{t_4})$ and $w\in V(\kket{\text{fail}}_{\mathfrak{d}_\mathfrak{lg}}^{t_4})$. So \eqref{eq:Q} on $_\mathfrak{lg}\!\ev*{\pi_{\neg\text{ok}_\mathfrak{d}}^{t_4}}{\text{fail}}_\mathfrak{lg}=1$ implies $\ev*{\mathcal{M}^{FR}, w}\models\square_\mathfrak{d}\neg\text{ok}^{t_4}_\mathfrak{d}$, where $\pi_{\neg\text{ok}_\mathfrak{d}}^{t_4}=\mathbb{I}-\pi_{\text{ok}}$. So $\ev*{\mathcal{M}^{FR}, w'}\models\neg\text{ok}^{t_4}_\mathfrak{d}$, for all $w'$ s.t.\ $wR_\mathfrak{d}w'$. But according to (I), $wR_\mathfrak{d}\hat{w}$. Hence,  $\hat{\mathcal{M}}^{FR}\models\text{ok}_\mathfrak{d}^{t_4}\wedge\neg\text{ok}^{t_4}_\mathfrak{d}$. \hfill $\blacksquare$
\end{proof}

The theorem is stronger since (I) \emph{proves} the existence of $\hat{w}$ and (II) shows that a contradiction follows already in serial Kripke-models. Ad hoc, this rules out a retreat to a doxastic interpretation of quantum probabilities as a quick fix. 

\section{Concerns and options}
In conjunction, Thm.s FR and FR* might indeed be seen as establishing the inadequacy of (epistemic/doxastic) modal logic in quantum settings. One might want to retreat to setting $R_\mathfrak{x}=\varnothing$ for all $\mathfrak{x}\in A$, as there would then be a Kripke model in which none of the two theorems follows. However, not only is \textbf{\textsf{K}} (in which this is possible) arguably too weak for modelling either knowledge or certainty: An empty access relation is incompatible with the non-vanishing probabilities implied by the \textsf{FR-Protocol} and condition \eqref{eq:Q}. What options are on the table for avoiding Nurgalieva and del Rio's conclusion that modal logic itself is the culprit?

One might react by assuming that quantum probabilities model something else than knowledge or certainty, maybe propensities. But this move would fail, for then probability-1 statements would model \emph{sure-fire} propensities and a correct application of QT would immediately \emph{give rise to} knowledge of outcomes.

Moreover, note that quantum probabilities as degrees of \emph{belief} are not confined to positions such as \emph{QBism} \citep{fuchs2014, mermin2014a, FMS2014} or related ones \citep{friederich2014, healey2017, boge2018}: several strategies for recovering the Born rule in many worlds interpretations have them quantify some kind of subjective uncertainty \citep[e.g.][]{saunders1998, sebens2016}. Similarly, known arguments for recovering the statistical content of QT in Bohmian approaches require the assumption that all particles in the universe be distributed according to the squared modulus of a global wave function $\Psi_0$ at some `initial time'. But Bohmian mechanics is perfectly deterministic, so the assumption of such a `distribution' amounts to ``providing a measure of subjective probability for the initial configuration'' \citep[45]{durr2012}. Hence, if we cannot interpret probability-1 assertions in QT either in terms of (or as giving rise to) knowledge or subjective certainty, this is bad news indeed.

Another response might be that an experimental realization of the \textsf{FR-Protocol} is inconceivable. For measurements involving human agents will typically not be of the ideal von Neumann-type, but rather have a non-negligible effect on the initial state of the measured object (human), i.e., proceed as $\ket{i}_\mathfrak{x}\ket{0}_\mathfrak{m}\mapsto\alpha_{if}\ket{f\neq i}_\mathfrak{x}\ket{i}_\mathfrak{m}$. However, the proof does not strictly require non-disturbing measurements, as inferences to pre-measurement states of agents and systems could still be drawn if the particular influence of the act of measuring was known. 

A more serious worry is that human agents are too complex to remain in \emph{coherent} states: repeated interactions with, say, air molecules and radiation will quickly decohere them into improper mixtures in a relevant basis.\footnote{Cf.\ \cite{schlosshauer2007} for an overview of decoherence theory.} In fact, Wigner himself later \cite{wigner1986} stepped back from his views on consciousness when he encountered Zeh's \cite{zeh1970} discovery of the decoherence mechanism.\footnote{Keep in mind, however, that decoherence \emph{by itself} does not solve the measurement problem, as so vividly pointed out by Bell \cite{bell1990}.}

These are clearly valid concerns, but one must be careful not to judge the experimental realizability of quantum peculiarities too quickly: Schr\"odinger \cite[848]{schrodinger}, for instance, suspected that a scenario of the EPR \citep[][]{epr1935} kind might be experimentally unrealisable; and today there is a flourishing field of Bell-type \citep{bell1964} experimentation.

Indeed, an experiment on the consequences of a related theorem by Brukner \cite{brukner} has recently been performed by Proietti et al. \cite{proietti}. However, Proietti et al. \cite[3]{proietti} ``define as observer any physical system that can extract information from another system by means of some interaction, and store that information in a physical memory'' and put \emph{photons} in the place of the two friends (Amanda and Gustavo) accordingly. Clearly, such an experiment has no bearing whatsoever on the issues discussed in this paper, as (neglecting panpsychism) photons are incapable of forming beliefs. Bottom line: at present, it is perfectly unclear whether the protocol could in principle ever be realized, but one should also not exclude this without a rigorous argument. 

What about principles \eqref{eq:Q}, \textsf{U}, or even \eqref{eq:tensor}? \eqref{eq:Q} was given a rather solid foundation on the basis of Halpern's work. However, in both proofs, we also proceeded \emph{from} knowledge/certainty \emph{to} a certain state of an agent, namely from $\square_\mathfrak{a}\ket*{+}^{t'}_\mathfrak{l}\ket*{0}^{t'}_\mathfrak{g}$ to $\kket{+,0}_{\mathfrak{a}_\mathfrak{lg}}^{t'}$. One might worry about the intrusion of circularity at this point, but this worry is easily dispersed: Condition \eqref{eq:Q} only applies in situations where the relevant state already pertains to the given agent and is otherwise silent. If we can, hence, establish some claim to knowledge/certainty of the relevant form independently of any direct quantum considerations, it seems perfectly sensible (almost undeniable) that an agent should be in a state that unambiguously indicates a certain state assignment. 

Similarly, denying \eqref{eq:tensor} is hardly feasible, for there does not seem to be anything better than product states \emph{within} the quantum formalism  for stating that one system is in one state and the other one in another. Note, moreover, that we have nowhere required that it always be possible to assert such a conjunction, so \eqref{eq:tensor} is quite weak. 

What about \textsf{U}? One way of denying \textsf{U} is Wigner's originally favored option: conscious cognitive agents are not \emph{just} physical systems, so we cannot always model them as evolving unitarily over time. This option introduces a serious mind-body dualism, a position not very prominent in present in philosophy of mind, for various convincing reasons. 

However, one need not go down this exact road in rejecting \textsf{U}. QBists or like interpreters of QT would, first off, certainly deny the adequacy of systems `being in states $\ket*{v}_\mathfrak{s}$', as quantum states are not seen as \emph{representational} in such interpretations: They need to be relativised to an \emph{assigner} therein, and one would have to introduce agent-dependent state assignments $\ket*{v}^{\mathfrak{x},t}_\mathfrak{s}$, meaning `$\mathfrak{x}$ assigns $\ket*{v}$ to $\mathfrak{s}$ at $t$'. 

This immediately obviates states of the form $\kket{v}_{\mathfrak{x}_\mathfrak{y}}$, but it does not block the two theorems: Given that she assigns a corresponding state to $\mathfrak{lg}$ at $t'$, Amanda could still evolve her Heisenberg operators unitarily and draw according inferences. Only if we denied, as suggested by Wigner, that QT could be applied in this way to other agents could we block the theorems. But with the representational function of quantum states discarded, the dualism hence invoked would be merely \emph{epistemic}: we need one set of rules for reasoning about what we construe as `inanimate matter' and another set for what we construe as `others'. None of this points \emph{directly} to a mind/ matter divide in the `real world', since our access to it is then remote, in part mediated by (non-representational) quantum states.

There are also some \emph{tacit} assumptions in both proofs, one being the establishment of claims to knowledge/ certainty of \emph{states}, not measurement results, on account of condition \eqref{eq:Q}. This move may seem innocent, but it is actually an invocation of the eigenvalue-eigenstate link \citep{fine1973}, as it proceeds directly from probabilities evaluated on Heisenberg operators to corresponding quantum states of systems. If this was denied, there would be no possibility for Amanda to infer anything about the state of $\mathfrak{lg}$ from contemplated measurements in the first place, and hence both Thm.\ FR and FR* would be blocked. 

This brings us to another tacit assumption, namely that merely contemplated, i.e., \emph{counterfactual} measurements can be treated on the same footing as actual ones. Indeed, Healey \cite[1577--9]{healey2018} argues that the proof \emph{erroneously} appeals to counterfactuals that concern dependencies between measurement-outcomes at different times and are maintained by an agent regardless of whether a measurement isolated from the measurements producing these outcomes was performed in between. 

In particular, Healey \cite[1586--8]{healey2018} targets a counterfactual leading from $\mathfrak{a}$'s measurement $1_\mathfrak{a}^{t_1}$ to fail$_\mathfrak{d}^{t_4}$, which is maintained even though $\mathfrak{c}$ measures $\mathfrak{ra}$ between the two events. Interestingly, the reason why Healey objects to this counterfactual is basically an acknowledgement of the stipulation $\mathfrak{s}\neq\mathfrak{x}$ in condition \textsf{U}: 
\begin{quote}
[$\mathfrak{a}$] is justified in using the state assignment [$\ket{+}_\mathfrak{l}^{t'}$] for the purpose of predicting the outcome of a measurement on [$\mathfrak{l}$] only where [$\mathfrak{l}$]'s correlations with other systems (encoded in an entangled state of a supersystem) may be neglected. \citep[1588]{healey2018}
\end{quote}

Hence, the fact that Chris and David can model Amanda herself as entangled with $\mathfrak{l}$ undercuts all inferences by Amanda involving the Heisenberg operators $\pi_{\text{fail}_\mathfrak{d}}^{t_4}, \pi_{\neg\text{ok}_\mathfrak{d}}^{t_4}$ in Thm.s FR and FR*.

This provides another, independent reason to be suspicious of  \textsf{U}, as Healey's arguments do not concern questions of consciousness but only \emph{correlations} that will be neglected if the locally established state ($\ket{+}_\mathfrak{l}^{t'}$) is taken to have general validity. Incidentally, a similar kind of argument could be made by a many worlds-theorist: The unitary $U_\mathfrak{a}$ has no meaning in a many worlds-view, as Amanda \emph{should} model her own evolution qua physical system by a unitary that affects her own state, and only this will provide the correct correlations that she, qua conscious being, will encounter after branching. Similarly, Bohmians could argue that Amanda would compute the wrong correlations between particles forming her brain and other systems if she used the local unitary $U_\mathfrak{a}$. Hence, \textsf{U}, as stated here and used by Frauchiger and Renner seems objectionable on several grounds.

\section{Conclusions}
This paper offered a thorough analysis of the theorem by Frauchiger and Renner within epistemic system \textsf{\textbf{T}} and provided an additional theorem compatible with a doxastic reading of quantum probability-1. Together, these theorems (FR, FR*) come dangerously close to justifying Nurgalieva and del Rio's assessment that modal logic might be inadequate for reasoning about credences or knowledge in quantum settings.

The subsequent discussion has shown, though, that there are several reasons for rejecting condition \textsf{U}, tacitly applied by Frauchiger and Renner and made explicit by Nurgalieva and del Rio and here: Many worlds-theorists, Bohmians, as well as Healeyan pragmatists should have reservations about the stipulation $\mathfrak{s}\neq \mathfrak{x}$ (more frankly: the isometries used by Frauchiger and Renner). 

Finally, the absence of iterated certainties of different agents in Thm.\ FR* defies Brukner's \citep[8--9]{brukner} assessment that:
\begin{quote}
the theorem by Frauchiger and Renner [...] points to the necessity to differentiate between ones' knowledge about direct observations and ones' knowledge about others' knowledge that is compatible with physical theories. It is likely that understanding this difference will be an important ingredient in further development of the method of Bayesian inference in situations as in the Wigner-friend experiment.
\end{quote}

However, as we have seen, QBist or like interpreters of QT could deny the applicability of quantum states and unitary QT to other epistemic  agents in the first place, without buying into Wignerian mind-body dualism. This would block Amanda's inferences in both theorems, and Brukner's remark could be maintained with a slight modification: The two theorems point to the necessity to differentiate between `states of matter' and `mental states of other agents', as conceived from the point of view of the agent using QT to make headway in her experienced environment.

\section*{Appendix}\label{sec:append}
\begin{A}\label{thm:OkZero}We show that ${}_{\mathfrak{ralg}}\!\bra{\text{init}}\Pi_{\neg(\text{ok}_\mathfrak{c} \wedge 0_\mathfrak{g})}\ket*{\text{init}}_\mathfrak{ralg}=1-0=1.$ For simplicity, we suppress system-indices where convenient and stick to the order $\mathfrak{ralg}$ for tensored quantum states. 

For reference, we note the unitaries that effect the state-transitions within steps 1 and 2 as given on the joint Hilbert space $\mathcal{H}_\mathfrak{ralg}$: 
\begin{align*}
U_{t_1} = & \pi_0\otimes\mathbb{I}_3 + \pi_1\otimes\sigma_x\otimes\mathbb{I}_2,\\
U_{t'} = & \pi_0\otimes\mathbb{I}_3 + \pi_1\otimes\mathbb{I} \otimes \sqrt{1/2}\left(\sigma_x + \sigma_z\right)\otimes\mathbb{I}, \\
U_{t_2} = & \pi_0\otimes\mathbb{I}_3 + \pi_1\otimes\mathbb{I}\otimes(\pi_0\otimes\mathbb{I} + \pi_1\otimes\sigma_x).
\end{align*}
Here, $\sigma_x = \dyad*{0}{1} +\dyad*{1}{0}$ and $\sigma_z= \dyad*{0}{0}-\dyad*{1}{1}$.\footnote{Recall also from the introduction that $\mathbb{I}$ acts on a single-system Hilbert space, $\mathbb{I}_2=\mathbb{I}\otimes\mathbb{I}$, $\mathbb{I}_3=\mathbb{I}_2\otimes\mathbb{I}$, etc.} In the Schr\"odinger picture, $U_{t_1}$ models the coupling of $\mathfrak{r}$ and $\mathfrak{a}$ in $\ket*{\text{init}}_\mathfrak{ralg}$ upon Amanda's first measurement, $U_{t'}$ models that both couple to $\mathfrak{l}$ in Amanda's preparation immediately after that, and $U_{t_2}$ $\mathfrak{g}$'s coupling with $\mathfrak{l}$ in his subsequent measurement. The reader may verify herself that these give the correct states and are, in fact, all not just unitary ($U_{i}^\dagger U_{i}= U_{i} U_{i}^\dagger =\mathbb{I}$) but also self-adjoint ($U_{i}^\dagger=U_{i}$). Note that we generally consider the evolution of local projectors extended to a larger space ($\pi\mapsto\Pi:=\mathbb{I}\otimes\ldots\otimes\mathbb{I}\otimes\pi\otimes\mathbb{I}\otimes\ldots$).

From  \eqref{eq:PiokPi0} it is straightforward to see that $\Pi_{\text{ok}}^{t_3}\Pi_{0}^{t_2}$ reduces to $\bar{U}^\dagger(\pi_\text{ok}\otimes\mathbb{I}_2)U_{t_2}(\mathbb{I}_2\otimes\pi_{0}\otimes\mathbb{I})U$. Hence  
\begin{align*}
&{}_{\mathfrak{ralg}}\!\bra{\text{init}}\Pi_{\text{ok}}^{t_3}\Pi_{0}^{t_2}\ket*{\text{init}}_\mathfrak{ralg}=\\
=&{}_{\mathfrak{ralg}}\!\bra{\Psi}(\pi_\text{ok}\otimes\mathbb{I}_2)U_{t_2}(\mathbb{I}_2\otimes\pi_{0}\otimes\mathbb{I})(\sqrt{1/3}\ket*{0}\ket*{0}\ket*{0}\ket*{0}+\sqrt{1/3}\ket*{1}\ket*{1}\ket*{0}\ket*{0}+ \sqrt{1/3}\ket*{1}\ket*{1}\ket*{1}\ket*{0})\\
=&{}_{\mathfrak{ralg}}\!\bra{\Psi}(\pi_\text{ok}\otimes\mathbb{I}_2)U_{t_2}(\sqrt{1/3}\ket*{0}\ket*{0}\ket*{0}\ket*{0}+\sqrt{1/3}\ket*{1}\ket*{1}\ket*{0}\ket*{0}) \\
=&{}_{\mathfrak{ralg}}\!\bra{\Psi}(\pi_\text{ok}\otimes\mathbb{I}_2)(\sqrt{1/3}\ket*{0}\ket*{0}\ket*{0}\ket*{0}+\sqrt{1/3}\ket*{1}\ket*{1}\ket*{0}\ket*{0})
\end{align*}
Since 
\begin{equation*}
\pi_\text{ok}=\dyad*{\text{ok}} = (1/2)(\dyad*{0}\otimes\dyad*{0} -\dyad*{0}{1}\otimes\dyad*{0}{1} + \dyad*{1}\otimes\dyad*{1} -\dyad*{1}{0}\otimes\dyad*{1}{0}),
\end{equation*}
it is easy to verify that the state to the right is an eigenvector of  $\pi_\text{ok}\otimes \mathbb{I}_2$ with eigenvalue zero, so we immediately get
\begin{equation*}
{}_{\mathfrak{ralg}}\!\bra{\text{init}}\Pi_{\text{ok}}^{t_3}\Pi_{0}^{t_2}\ket*{\text{init}}_\mathfrak{ralg}= 0.
\end{equation*}
\hfill $\blacksquare$
\end{A}
\begin{A}\label{thm:JointOcc}
We show that $1-{}_\mathfrak{ralg}\!\ev{\Pi_\text{ok}^{t_4}\Pi_\text{ok}^{t_3}\Pi_1^{t_2}\Pi_1^{t_1}}{\text{init}}_\mathfrak{ralg} = 1$. Until $\mathfrak{a}$'s measurement is completed at $t_1$, the joint state of $\mathfrak{r}$ and $\mathfrak{a}$'s memory will evolve trivially, whence 
\begin{equation*}
\Pi_1^{t_1} :=  \mathbb{I}_4(\pi_1\otimes\mathbb{I}_{3})\mathbb{I}_4 =\pi_1\otimes\mathbb{I}_{3}.
\end{equation*}
After that, the global state $\ket*{\text{init}}_\mathfrak{ralg}$ will evolve under $U_{t_1}$ and $U_{t'}$ in the Schr\"odinger picture, as described in \hyperref[thm:OKZero]{A 1}. Accordingly, 
\begin{equation*}\label{eq:P2P1}
\Pi_1^{t_2} :=  U_{t_1} U_{t'}(\mathbb{I}_2\otimes\pi_1\otimes\mathbb{I})U_{t'} U_{t_1} 
\end{equation*}

As for $\Pi_\text{ok}^{t_3}$ and $\Pi_\text{ok}^{t_4}$, each of  these needs to be evolved by all unitaries $U_{i}$, whereafter the evolution is again considered trivial on all spaces. Since the term $U_{t_2}U_{t'}U_{t_1}U_{t_1}U_{t'}U_{t_2}$ in the middle of $\Pi_\text{ok}^{t_4}\Pi_\text{ok}^{t_3}$ gives the identity and the untransformed projectors will be of the form $\mathbb{I}_2\otimes\pi_\text{ok}, \pi_\text{ok}\otimes\mathbb{I}_2$, we only need to consider 
\begin{equation*}\label{eq:PokPok}
U_{t_1}U_{t'}U_{t_2}\pi_\text{ok}\otimes\pi_\text{ok}U_{t_2}U_{t'}U_{t_1}.
\end{equation*}

Now we can let $U_{t_1}U_{t'}U_{t_2}$ act on ${}_\mathfrak{ralg}\!\bra*{\text{init}}$ to obtain ${}_\mathfrak{ralg}\!\bra*{\Psi}$, as in eqs.\ \eqref{eq:Chris} and \eqref{eq:David}. Moreover, we find that $U_{t'}U_{t_1}U_{t_1}U_{t'}=\mathbb{I}_4$ in between  $\Pi_\text{ok}^{t_4}\Pi_\text{ok}^{t_3}$ and $\Pi_1^{t_2}\Pi_1^{t_1}$. Therefore, we need to evaluate 
\begin{align*}
&{}_\mathfrak{ralg}\!\bra*{\Psi}(\pi_\text{ok}\otimes\pi_\text{ok})U_{t_2}(\mathbb{I}_2\otimes\pi_1\otimes\mathbb{I})U_{t'}U_{t_1}(\pi_1\otimes\mathbb{I}_{3})\ket*{\text{init}}_\mathfrak{ralg} = \\
= & \sqrt{1/12}\bra*{\text{ok}}\bra*{\text{ok}}U_{t_2}(\mathbb{I}_2\otimes\pi_1\otimes\mathbb{I})U_{t'}U_{t_1}(\pi_1\otimes\mathbb{I}_{3})\ket*{\text{init}}_\mathfrak{ralg} = \\
=&\sqrt{1/24}\bra*{\text{ok}}\bra*{\text{ok}}U_{t_2}(\mathbb{I}_2\otimes\pi_1\otimes\mathbb{I})(\pi_0\otimes\mathbb{I}_3 + \pi_1\otimes\sigma_x \otimes\left(\sigma_x + \sigma_z\right)\otimes\mathbb{I})(\pi_1\otimes\mathbb{I}_{3})\ket*{\text{init}}_\mathfrak{ralg}=\\
=&\sqrt{1/24}\bra*{\text{ok}}\bra*{\text{ok}}U_{t_2}(\pi_1\otimes\sigma_x \otimes\pi_1\left(\sigma_x + \sigma_z\right)\otimes\mathbb{I})\ket*{\text{init}}_\mathfrak{ralg}=\\
=&\sqrt{1/24}\bra*{\text{ok}}\bra*{\text{ok}}U_{t_2}(\pi_1\otimes\sigma_x \otimes(\dyad*{1}{0}-\pi_1)\otimes\mathbb{I}\ket*{\text{init}}_\mathfrak{ralg}=\\
=&\sqrt{1/24}\bra*{\text{ok}}\bra*{\text{ok}}\pi_1\otimes\sigma_x \otimes(\pi_0\otimes\mathbb{I} + \pi_1\otimes\sigma_x)(\dyad*{1}{0}\otimes\mathbb{I}-\pi_1\otimes\mathbb{I})\ket*{\text{init}}_\mathfrak{ralg}\\
=&\sqrt{1/24}\bra*{\text{ok}}\bra*{\text{ok}}\pi_1\otimes\sigma_x \otimes(\dyad*{1}{0}-\pi_1)\otimes\sigma_x)\ket*{\text{init}}_\mathfrak{ralg}.
\end{align*} 
Consider the two operators $A=\pi_1\otimes\sigma_x \otimes\dyad*{1}{0}\otimes\sigma_x$ and $B=-\pi_1\otimes\sigma_x\otimes\pi_1\otimes\sigma_x$ on $\ket*{\text{init}}_\mathfrak{ralg}=\sqrt{1/3}\ket*{0}\ket*{0} \ket*{0}\ket*{0} +\sqrt{2/3}\ket*{1}\ket*{0} \ket*{0}\ket*{0}$ individually. Due to $\pi_1$'s appearance in places other than the first, $\ket*{\text{init}}_\mathfrak{ralg}$ is an eigenvector of  $B$  with eigenvalue 0. $A$ eliminates the first term but flips all 0s in the second term of $\ket*{\text{init}}_\mathfrak{ralg}$ to 1s. Hence we get  
\begin{align*}
&{}_\mathfrak{ralg}\!\bra*{\Psi}(\pi_\text{ok}\otimes\pi_\text{ok})U_{t_2}(\mathbb{I}_2\otimes\pi_1\otimes\mathbb{I})U_{t'}U_{t_1}(\pi_1\otimes\mathbb{I}_{3})\ket*{\text{init}}_\mathfrak{ralg} = \\
=& (1/6)(\bra*{\text{ok}}\bra*{\text{ok}})(\ket*{1}\ket*{1}\ket*{1}\ket*{1}) = 1/12.
\end{align*} 
Note that this implies $\Pr(\text{ok}_{\mathfrak{c}}^{t_3}\wedge \text{ok}_{\mathfrak{d}}^{t_4}|1_{\mathfrak{a}}^{t_1}\wedge 1_{\mathfrak{g}}^{t_2})=1$.
\hfill $\blacksquare$
\end{A}
\begin{A}\label{thm:AmandGust}
We show that ${}_\mathfrak{ralg}\!\bra*{\text{init}}\Pi_{\neg(1_\mathfrak{g}\wedge 0_\mathfrak{a})}\ket*{\text{init}}_\mathfrak{ralg} =1$. To do so, we only need to show that ${}_\mathfrak{ralg}\!\bra*{\text{init}}\Pi^{t_2}_1\Pi_0^{t_1}\ket*{\text{init}}_\mathfrak{ralg} =0$, where
\begin{equation*}
\Pi_0^{t_1} := \pi_0\otimes\mathbb{I}_{3},
\end{equation*}
and $\Pi^{t_2}_1$ is defined as in \hyperref[thm:JointOcc]{A 2}. 

Inserting these, we get
\begin{equation*}
{}_\mathfrak{ralg}\!\bra*{\text{init}}\Pi^{t_2}_1\Pi_0^{t_1}\ket*{\text{init}}_\mathfrak{ralg} ={}_\mathfrak{ralg}\!\bra*{\text{init}}U_{t_1}U_{t'}(\mathbb{I}_2\otimes\pi_1\otimes\mathbb{I})U_{t'}U_{t_1}(\pi_0\otimes\mathbb{I}_{3})\ket*{\text{init}}_\mathfrak{ralg}.
\end{equation*}
It is easy to check that $U_{t_1}U_{t'} = U_{t'}U_{t_1} = (\pi_0\otimes\mathbb{I}_3 + \pi_1\otimes\sigma_x \otimes\left(\sigma_x + \sigma_z\right)\otimes\mathbb{I})$. Inserting this, we get
\begin{align*}
&{}_\mathfrak{ralg}\!\bra*{\text{init}}\Pi^{t_2}_1\Pi_0^{t_1}\ket*{\text{init}}_\mathfrak{ralg} =\\
=&{}_\mathfrak{ralg}\!\bra*{\text{init}}(\pi_0\otimes\mathbb{I}\otimes\pi_1\otimes\mathbb{I} + \pi_1\otimes\sigma_x \otimes\left(\dyad*{0}{1}-\dyad*{1}{1} \right)\otimes\mathbb{I})(\pi_0\otimes\mathbb{I}_3)\ket*{\text{init}}_\mathfrak{ralg}/2=\\
=&{}_\mathfrak{ralg}\!\bra*{\text{init}}\pi_0\otimes\mathbb{I}\otimes\pi_1\otimes\mathbb{I}\ket*{\text{init}}_\mathfrak{ralg}/2.
\end{align*}
Now $\ket*{\text{init}}_\mathfrak{ralg}=\sqrt{1/3}\ket*{0}\ket*{0}\ket*{0}\ket*{0}+\sqrt{2/3}\ket*{1}\ket*{0}\ket*{0}\ket*{0})$, so  $\ket*{\text{init}}_\mathfrak{ralg}$ is an eigenvector of this operator with eigenvalue zero.
\hfill $\blacksquare$
\end{A}
\begin{A}\label{thm:Amand}
We show that ${}_\mathfrak{l}\!\bra*{+} {}_\mathfrak{g}\!\bra*{0}\pi_\text{fail}^{t_4}\ket*{+}_\mathfrak{l}\ket*{0}_\mathfrak{g}={}_\mathfrak{l}\!\bra*{+} {}_\mathfrak{g}\!\bra*{0}\pi_{\neg\text{ok}}^{t_4}\ket*{+}_\mathfrak{l}\ket*{0}_\mathfrak{g}=1$, where $\pi_\text{fail}^{t_4}=U^{\dagger}_\mathfrak{a}\pi_\text{fail}U_\mathfrak{a}$, $\pi_{\neg\text{ok}}^{t_4}=U^{\dagger}_\mathfrak{a}(\mathbb{I}_{2}-\pi_\text{ok})U_\mathfrak{a}$, and we recall that $U_\mathfrak{a}= \pi_0\otimes\mathbb{I} + \pi_1\otimes\sigma_x$ is the $\mathfrak{lg}$-part of $U_{t_2}$, postselected for $1_\mathfrak{a}^{t_1}$. Using the Schr\"odinger picture, it is easy to verify that 
\begin{equation*}
U_\mathfrak{a}\ket*{+}_\mathfrak{l}\ket*{0}_\mathfrak{g}=\ket*{\text{fail}}_\mathfrak{lg},
\end{equation*}
so it is immediate that 
\begin{equation}
_\mathfrak{l}\!\bra*{+} {}_\mathfrak{g}\!\bra*{0}\pi_\text{fail}^{t_4}\ket*{+}_\mathfrak{l}\ket*{0}_\mathfrak{g}={}_\mathfrak{lg}\!\bra*{\text{fail}}\pi_\text{fail}\ket*{\text{fail}}_\mathfrak{lg}=1.
\end{equation}
Moreover, 
\begin{equation*}
\braket*{\text{ok}}{\text{fail}}=(\bra*{0}\bra*{0}-\bra*{1}\bra*{1})(\ket*{0}\ket*{0}+\ket*{1}\ket*{1})/2=(1-1)/2=0,
\end{equation*}
which establishes 
\begin{equation*}
{}_\mathfrak{l}\!\bra*{+} {}_\mathfrak{g}\!\bra*{0}\pi_{\neg\text{ok}}^{t_4}\ket*{+}_\mathfrak{l}\ket*{0}_\mathfrak{g}=1 - {}_\mathfrak{lg}\!\bra*{\text{fail}}\pi_\text{ok}\ket*{\text{fail}}_\mathfrak{lg}=1-0=1.
\end{equation*}

\hfill $\blacksquare$
\end{A}
\begin{acknowledgements}I have profited from discussions with Chris Fuchs and David Glick at the Stellenbosch Institute for Advanced Study, and I thank Chris for the invitation and David for pointing him to my work. Otherwise, I thank Markus Schrenk, Dennis Lehmkuhl, and Cristin Chall for helpful discussion and comments. \end{acknowledgements}
\bibliographystyle{abbrv}
\bibliography{qcert}
\end{document}